\documentclass[journal,onecolumn]{IEEEtran}

\usepackage{amsmath, amsthm, amssymb}
\usepackage{latexsym}
\usepackage{graphicx}
\usepackage{verbatim}
\usepackage{mathrsfs}
\usepackage{float}
\usepackage[lofdepth, lotdepth]{subfig}
\usepackage{array}
\usepackage{url}
\usepackage{algorithm}
\usepackage[noend]{algorithmic}
\usepackage{cite}
\usepackage[table]{xcolor}
\usepackage{enumerate}
\usepackage{eucal}
\usepackage{stmaryrd}
\usepackage{mathtools}
\usepackage{pgf}
\usepackage{tikz}
\usetikzlibrary{arrows,automata}
\usepackage[latin1]{inputenc}
\usetikzlibrary{automata,positioning}
\usepackage[margin=0.65in]{geometry}
\usepackage{mathtools}
\DeclarePairedDelimiter\ceil{\lceil}{\rceil}

\theoremstyle{plain}
\newtheorem{theorem}{Theorem}
\newtheorem{proposition}[theorem]{Proposition}
\newtheorem{lemma}[theorem]{Lemma}

\newtheorem{construction}{Construction}
\theoremstyle{definition}
\newtheorem{definition}[theorem]{Definition}
\newtheorem{example}[theorem]{Example}
\newtheorem{remark}[theorem]{Remark}

\newcommand{\B}{{\mathcal B}}
\newcommand{\C}{{\mathcal C}}
\newcommand{\D}{{\mathcal D}}

\DeclareMathAlphabet{\mathbfsl}{OT1}{ppl}{b}{it} 

\newcommand{\bM}{{\mathbfsl{M}}}

\newcommand{\bL}{\mathbfsl{L}}

\newcommand{\bU}{\mathbfsl{U}}
\newcommand{\bW}{\mathbfsl{W}}
\newcommand{\bR}{\mathbfsl{R}}

\newcommand{\bu}{{\mathbfsl u}}

\newcommand{\by}{{\mathbfsl y}}

\newcommand{\bd}{{\mathbfsl{d}}}
\newcommand{\bk}{{\mathbfsl{k}}}

\newcommand{\bc}{{\mathbfsl c}}

\newcommand{\bx}{{\mathbfsl{x}}}
\newcommand{\bz}{{\mathbfsl{z}}}

\newcommand{\bsg}{{\boldsymbol{\sigma}}}

\newcommand{\bbJ}{{\mathbb J}}

\newcommand{\bbZ}{{\mathbb Z}}

\newcommand{\ppmod}[1]{~({\rm mod~}#1)}

\renewcommand{\ge}{\geqslant}
\renewcommand{\le}{\leqslant}

\newcommand{\et}{{\emph{et al.}}}

\newcommand{\tA}{\mathtt{A}}
\newcommand{\tT}{\mathtt{T}}
\newcommand{\tC}{\mathtt{C}}
\newcommand{\tG}{\mathtt{G}}
\newcommand{\GC}{$\mathtt{GC}$}

\newcommand{\enc}{\textsc{Enc}}
\newcommand{\dec}{\textsc{Dec}}

\newcommand{\Lburst}{{\rm L}^{\rm burst}}

\newcommand{\Bindel}{{\cal B}^{\rm indel}}
\newcommand{\Bedit}{{\cal B}^{\rm edit}}

\newcommand{\bburst}{b{\rm -burst}}
\newcommand{\rsb}[1]{#1{\rm -rsb}}

\begin{document}

\pagestyle{empty}

\title{Optimal Codes Correcting a Single Indel / Edit \\ for DNA-Based Data Storage} 
\author{
   \IEEEauthorblockN{
   	Kui Cai,
	Yeow Meng Chee, 
	Ryan Gabrys,
	Han Mao Kiah,
	and Tuan Thanh Nguyen}
\thanks{Yeow Meng Chee is with the Department of Industrial Systems Engineering and Management, National University of Singapore, Singapore 119077 (email: pvocym@nus.edu.sg).}
\thanks{Ryan Gabrys is with the Spawar Systems Center, San Diego, CA 92152 USA (email: ryan.gabrys@gmail.com).}
\thanks{Han Mao Kiah is with the School of Physical and Mathematical Sciences, Nanyang Technological University, Singapore 637371 (email: hmkiah@ntu.edu.sg).}
\thanks{Kui Cai and Tuan Thanh Nguyen are with the Singapore University of Technology and Design, Singapore 487372 (email: \{cai\_kui, tuanthanh\_nguyen\}@sutd.edu.sg).}
\thanks{This paper was presented in part at 2019 IEEE International Symposium on Information Theory \cite{chee2019edit}.}
}

\maketitle


\begin{abstract}
An indel refers to a single insertion  or deletion, 
while an edit refers to a single insertion, deletion or substitution.
In this paper, we investigate codes that combat either a single indel or a single edit and 
provide linear-time algorithms that encode binary messages into these codes of length $n$.
Over the quaternary alphabet, we provide two linear-time encoders.
One corrects a single edit with $\ceil{\log n}+O(\log \log n)$ redundancy bits, while
the other corrects a single indel with $\ceil{\log n}+2$ redundant bits. 
These two encoders are {\em order-optimal}.
The former encoder is the first known order-optimal encoder that corrects a single edit, 
while the latter encoder (that corrects a single indel) reduces the redundancy of the best known encoder of Tenengolts (1984) by at least four bits.
Over the DNA alphabet, we impose an additional constraint:  the {\em \GC-balanced constraint} and 
require that exactly half of the symbols of any DNA codeword to be either $\tC$ or $\tG$.
In particular, via a modification of Knuth's balancing technique, we provide a linear-time map that translates binary messages into \GC-balanced codewords and the resulting codebook is able to correct a single indel or a single edit.
These are the first known constructions of \GC-balanced codes that correct a single indel or a single edit.
\end{abstract}


\section{Introduction}

Advances in synthesis and sequencing technologies 
have made DNA macromolecules an attractive medium for digital information storage.
Besides being biochemically robust, DNA strands offer ultrahigh storage densities of $10^{15}$-$10^{20}$ bytes per gram of DNA, as demonstrated in recent experiments (see \cite[Table 1]{Yazdi.2017}). 
These synthetic DNA strands may be stored {\em ex vivo} or {\em in vivo}. 
When the DNA strands are stored {\em ex vivo} or in a non-biological environment, 
code design takes into account the synthesising and sequencing platforms being used 
(see \cite{Yazdi.etal:2015b} for a survey).
In contrast, when the DNA strands are stored {\em in vivo} or recombined with the DNA of a living organism,
we design codes to correct errors due to the biological mutations \cite{Jain:2017}. 

Common to both environments are errors due to {\em insertion}, {\em deletion} and {\em substitution}.
For example, Organick \et{} recently stored 200MB of data in 13 million DNA strands and 
reported insertion, deletion and substitution rates to be $5.4\times 10^{-4}$, $1.5\times 10^{-3}$ and $4.5\times 10^{-3}$, respectively \cite{Organick:2018}.
When DNA strands are stored {\em in vivo}, these errors are collectively termed as {\em point mutations} and occur during the process of DNA replication \cite{Clancy:2008}.

For convenience, we refer to a single insertion  or deletion as an {\em indel}, 
and a single insertion, deletion or substitution as an {\em edit}. 
In this work, we investigate codes that combat either a single indel or a single edit and 
provide efficient methods of encoding binary messages into these codes.

To correct a single indel, we have the celebrated class of Varshamov-Tenengolts (VT) codes. 
While Varshamov and Tenengolts introduced the binary version to correct asymmetric errors \cite{var1965},
Levenshtein later provided a linear-time decoding algorithm to correct a single indel for the VT codes \cite{le1965}.
In the same paper, Levenshtein modified the VT construction to correct a single edit.
In both constructions, the number of redundant bits is $\log n+O(1)$, where $n$ is the length of a codeword.
Unless otherwise stated, all logarithms are taken base two.

The VT construction partitions all binary words of length $n$ into certain $n+1$ classes, where each class is a VT code.
Curiously, even though efficient decoding  was known since 1965,
a linear-time algorithm to encode into a VT code was only proposed by Abdel-Ghaffar and Ferriera in 1998
\cite{kas1998}.

A nonbinary version of the VT codes was proposed by Tenengolts \cite{te1984},
who also provided a linear-time method to correct a single indel.
In the same paper, Tenengolts also provided an efficient encoder that corrects a single indel.
For the quaternary alphabet, this encoder requires at least $\log n + 7$ bits for words of length $n\ge 20$.
However, the codewords obtained from this encoder are not contained in a single non-binary VT code 
(see Section~\ref{sec:previous} for a discussion).
Hence, recently, Abroshan \et{} presented a systematic encoder that maps words into a single non-binary VT code \cite{abroshan}. 
Unfortunately,  the redundancy of this encoder is $\lceil\log n\rceil (\log q +1)+2(\log q-1)$, and when $q=4$, 
the redundancy is $3\lceil\log n\rceil+2$.
To the best of our knowledge, 
there is no known efficient construction for $q$-ary codes (or even quaternary codes) that can correct a single edit.

To further reduce errors, we also impose certain weight constraints on the individual codewords.
Specifically, the {\em \GC-content} of a DNA string refers to the percentage of nucleotides that corresponds to $\tG$ or $\tC$,
and DNA strings with \GC-content that are too high or too low are more prone to both synthesis and sequencing errors (see for example, \cite{Yakovchuk:2006, ross2013}).
Therefore, most work use DNA strings whose \GC-content is close to 50\% as codewords
and use randomizing techniques to encode binary message into the latter \cite{Organick:2018}.
Recently, in addition to the \GC-content constraints, Immink and Cai studied the homopolymer runlength constraint for DNA codewords \cite{kee2018}.

In this paper, in addition to correcting either a single indel or a single edit, 
we provide linear-time encoders that map binary messages into codewords that have \GC-content exactly 50\%.
To the best of our knowledge, no such codebooks are known prior to this work. 
We summarize our contributions.
\begin{enumerate}[(A)]
\item In Section~\ref{sec:indel}, we present a linear-time quaternary encoder that corrects a single indel with $\ceil{\log n}+2$ bits of redundancy. This construction improves the encoder of Tenengolts \cite{te1984} by reducing the redundancy by {\em at least four bits}. We then proceed to extend and generalize this construction so as to design efficient encoder for codes capable of correcting a burst of indels with $\log n+O(\log \log n)$ bits of redundancy. 

\item {
In Section~\ref{sec:edit}, we construct two classes of quaternary codes that corrects a single edit.
The first class of codes incurs $2\ceil{\log n}+2$ bits of redundancy, 
while the second class of codes incurs only $\log n+O(\log \log n)$ bits of redundancy and is thus order-optimal. 
Even though the former is not order-optimal, it outperforms the latter class when $n\le 512$.
In Section~\ref{sec:ntedit}, we study a type of edits specific to the DNA storage channel and 
provide an efficient encoder that correct such edits with $\log n+\log \log n + O(1)$ bits of redundancy.
}

\item In Section~\ref{sec:balanced}, we encode binary messages to \GC-balanced codewords.
Via a modification of Knuth's balancing method, we obtain \GC-balanced single indel/edit-correcting encoders.
\end{enumerate}

We first go through certain notation and define the problem. 
{For the convenience of the reader, relevant notation and terminology referred to throughout the paper is summarized in Table I.}

\begin{table}
\renewcommand{\arraystretch}{1.3}
\begin{tabular}{ p{5.8cm} p{11.8cm} }
 \hline
 Notation    & Description \\
 \hline
 $\Sigma$& alphabet of size $q$ \\
 $\Sigma_4$ &quaternary alphabet, i.e. $q=4, \Sigma_4=\{0,1,2,3\}$\\
 $\D$& DNA alphabet, $\D=\{\tA, \tT, \tC, \tG\}$\\
 $\bx \by$& the concatenation of two sequences \\
 $\bx || \by$ & the interleaved sequence \\
 $\bsg, \bU_\bsg, \bL_\bsg$ & a DNA sequence $\bsg$, the upper sequence of $\bsg$, and the lower sequence of $\bsg$ \\
  $\Psi$  & the one-to-one map that converts a DNA sequence to a binary sequence \\
  ${\rm Syn}(\bx)$& the syndrome of a sequence $\bx$\\
  ${\rm Rsyn}(\bx)$& the run-syndrome of a sequence $\bx$\\
  ${\rm indel} $ & single insertion or single deletion \\
  ${\rm edit} $ & single insertion, or single deletion, or single substitution\\
  $b$-burst-indel & $b$ consecutive insertions or $b$ consecutive deletions \\
  $b$-burst-edit & $b$ consecutive insertions, or $b$ consecutive deletions, or $b$ consecutive substitutions\\
  $\Bindel(\bx)$   & the set of words that can be obtained from $\bx$ via at most a single indel \\
  $\Bedit(\bx)$    &  the set of words that can be obtained from $\bx$ via at most a single edit\\
 $\Bindel_{\bburst}(\bx)$      & the set of words that can be obtained from $\bx$ via a $b$-burst-indel \\
  $\Bedit_{\bburst}(\bx)$      & the set of words that can be obtained from $\bx$ via a $b$-burst-edit\\
  ${\rm Bal}_k(n)$& the set of quaternary words of length $n$ that is $k$-sum-balanced\\
\hline
\end{tabular}
\vspace{5mm}

\begin{tabular}{ p{3cm} p{9.3cm} p{3cm} p{1.5cm}}
\hline
  Code and \newline Encoder / Decoder   & Description & Redundancy & Remark \\
 \hline
  ${\rm VT}_a(n)$& 
  binary Varshamov-Tenengolts code that corrects a single indel & &Section II-B\\
  
  ${\rm L}_a(n)$& binary Levenshtein code that corrects a single edit & &Section II-B\\
  
  $\enc_{\mathbb{L}}, \dec_{\mathbb{L}}$& encoder and decoder for ${\rm L}_a(n)$ &$\ceil{\log n}+1$ bits &Section II-C\\
  
  ${\rm T}_{a,b}({n;q})$& nonbinary VT code that corrects a single indel & &Section II-B\\
  \hline
  
  $\Lburst_a(n)$ & binary Levenshtein code that corrects a 2-burst-indel & &Section III-A\\
  
  $\C_a(n)$& DNA code that corrects a single indel & &Section III-A\\
  
  $\enc_{\mathbb{I}}, \dec_{\mathbb{I}}$& encoder and decoder for $\C_a(n)$ & $\ceil{\log n}+2$ bits &Section III-B\\
    ${\rm SVT}_{c,d,P}(n)$& Shifted-VT code that correct a single indel/edit provided that the error is located within $P$ consecutive positions & &Section III-C\\
    
  $\enc_{SVT}$& encoder for ${\rm SVT}_{c,d,P}(n)$ & $\ceil{\log P}+2$ bits &Section III-C\\
  $\enc_{RLL}$& encoder that outputs a binary sequence that the longest run is at most $\ceil{\log n}+3$ & one bit &Section III-C\\
  $\C_{\bburst}^{\rm indel}(n,P,r;a,c,d)$& DNA code that corrects a $b$-burst-indel & &Section III-C\\
  $\enc_{\bburst}^{\rm indel}$& encoder for $\C_{\bburst}^{indel}(n,P,r;a,c,d)$ & $\log n+O(\log \log n)$ bits &Section III-C\\
  \hline
  
  $\enc_{\mathbb{E}}^{A}$& encoder for DNA codes that corrects a single edit &$2\ceil{\log n}+2$ bits  &Section IV-A\\
  $\C^B(n;a,b,c,d)$& order-optimal DNA code that corrects a single general edit & &Section IV-B\\
  $\C_{a,b,c}^{\rm nt}(n,r,P)$& order-optimal DNA code that corrects a single nucleotide edit & &Section IV-C\\
  $\enc_{\mathbb{E}}^{{\rm nt}}, \dec_{\mathbb{E}}^{{\rm nt}}$& the encoder and decoder for $\C_{a,b,c}^{\rm nt}(n,r,P)$ & $\log n+O(\log \log n)$ bits &Section IV-C\\
  \hline
  $\enc_{\tG\tC}, \dec_{\tG\tC}$
  & encoder and decoder for GC-balanced DNA code that corrects a single edit & $3\ceil{\log n}+2$ bits &Section V\\
  \hline
 \end{tabular}
\caption{Notation Summary}\label{tab.notation}
\end{table}


\section{Preliminary}\label{sec:prelim}

Let $\Sigma$ denote an {\em alphabet} of size $q$.
For any positive integer $m<n$, we let $[m,n]$ denote the set $\{m,m+1,\ldots,n\}$ and $[n]=[1,n]$.

Given two sequences $\bx$ and $\by$, we let $\bx\by$ denote the {\em concatenation} of the two sequences.
In the special case where $\bx,\by \in \Sigma^n$, we use $\bx || \by$ to denote their {\em interleaved sequence} $x_1y_1x_2y_2\ldots x_ny_n$.
For a subset $I=\{i_1,i_2,\ldots, i_k\}$ of coordinates, we use $\bx|_I$ to denote the vector $x_{i_1}x_{i_2}\ldots x_{i_k}$.

Let $\bx\in \Sigma^n$. We are interested in the following {\em error balls}:
\begin{align*}
\Bindel(\bx) & \triangleq \{\bx\}\cup\{ \by : \by \mbox{ is obtained from $\bx$ via a single insertion or deletion}\},\\
\Bedit(\bx)   & \triangleq \{\bx\}\cup\{ \by : \by \mbox{ is obtained from $\bx$ via a single insertion, deletion or substitution}\}.
\end{align*}
Observe that when $\bx\in\Sigma^n$, both $\Bindel(\bx)$ and $\Bedit(\bx)$ are subsets of $\Sigma^{n-1}\cup \Sigma^n \cup \Sigma^{n+1}$.
Hence, for convenience, we use $\Sigma^{n*}$ to denote the set $\Sigma^{n-1}\cup \Sigma^n \cup\Sigma^{n+1}$.

Let $\C\subseteq \Sigma^n$. We say that {\em $\C$ corrects a single indel} if 
$\Bindel(\bx)\cap \Bindel(\by)=\varnothing$ for all distinct $\bx,\by \in \C$.
Similarly, {\em $\C$ corrects a single edit} if 
$\Bedit(\bx)\cap \Bedit(\by)=\varnothing$ for all distinct $\bx,\by \in \C$.
In this work, not only are we interested in constructing large codes that correct a single indel or edit,
we desire efficient encoders that map binary messages into these codes.

\begin{definition}
The map $\enc: \{0,1\}^m\to \Sigma^n$
is a {\em single-indel-correcting encoder} if there exists a {\em decoder} map $\dec:\Sigma^{n*} \to \{0,1\}^m \cup \{?\}$ such that the following hold.
\begin{enumerate}[(i)]
\item For all $\bx\in\{0,1\}^m$, we have $\dec\circ\enc(\bx)=\bx$.
\item If $\bc=\enc(\bx)$ and $\by\in \Bindel(\bc)$, then $\dec(\by)=\bx$.
\end{enumerate}
Hence, we have that the code $\C=\{\bc : \bc=\enc(\bx),\, \bx\in\{0,1\}^m\}$ corrects a single indel and $|\C|=2^m$.
The {\em redundancy of the encoder} is measured by the value $n \log q -  m$ (in bits).
A  {\em single-edit-correcting encoder} is similarly defined.
\end{definition}

Therefore, our design objectives for a single-indel-correcting or single-edit-correcting encoder are as follow.
\begin{itemize}
\item The redundancy is $K \log n + o(\log n)$, where $K$ is a constant to be minimized.
When $K=1$, we say that the encoder is {\em order-optimal}.
\item The encoder $\enc$ can be computed in time $O(n)$.
\item The decoder $\dec$ can be computed in time $O(n)$.
\end{itemize}


\subsection{DNA Alphabet}

When $q=4$, we denote the alphabet by $\D=\{\tA, \tT, \tC, \tG\}$
and consider the following one-to-one correspondence between $\D$ and two-bit sequences:
\[ \tA \leftrightarrow 00,\quad \tT \leftrightarrow 01,\quad\tC \leftrightarrow 10,\quad\tG \leftrightarrow 11.\]
Therefore, given a sequence $\bsg\in \D^n$, we have a corresponding a binary sequence $\bx\in\{0,1\}^{2n}$
and we write $\bx=\Psi(\bsg)$.

Let $n$ be even. 
We say that $\bsg\in \D^n$ is \emph{ \GC-balanced} if 
the number of symbols in $\bsg$ that correspond to $\tC$ and $\tG$ is $n/2$.
On the other hand, we that that $\bx\in \{0,1\}^n$ is {\em balanced}
if the number of ones in $\bx$ is $n/2$.
For DNA-based storage, we are interested in codewords that are \GC-balanced.

\begin{definition}
A single-indel-correcting encoder $\enc:\{0,1\}^m \to\D^n$ is a {\em \GC-balanced single-indel-correcting encoder} 
if $\enc(\bx)$ is \GC-balanced for all $\bx\in\{0,1\}^m$.
A {\em \GC-balanced single-edit-correcting encoder} is similarly defined.
\end{definition}

Given $\bsg \in \D^n$, let $\bx=\Psi(\bsg)\in \{0,1\}^{2n}$
and we set $\bU_\bsg=x_1x_3\cdots x_{2n-1}$ and $\bL_\bsg=x_2x_4\cdots x_{2n}$.
In other words, $\bsg=\Psi^{-1}(\bU_\bsg || \bL_\bsg)$.
We refer to $\bU_{\sigma}$ and $\bL_\bsg$ as the {\em upper sequence} and {\em lower sequence} of $\bsg$, respectively.

The following example demonstrates the relation between $\bsg$, $\bU_{\bsg}$ and $\bL_{\sigma}$.

\begin{example}\label{example1}
Suppose that $\bsg=\tA\tC\tA\tG\tT\tG$ and we check that $\bsg$ is \GC-balanced. 
Now, $\bx\triangleq \Psi(\bsg)=001000110111$ and we write $\bU_\bsg$ and $\bL_\bsg$ as follow.

\begin{center}
\begin{tabular}{|c|| c | c| c| c| c|c|}
 \hline
 $\bsg$ & $\tA$ & $\tC$ & $\tA$  & $\tG$ & $\tT$ & $\tG$ \\ \hline
 $\bU_{\bsg}$ & 0 & 1 & 0 & 1 & 0 & 1 \\\hline
 $\bL_{\bsg}$ & 0 & 0 & 0 & 1 & 1 & 1\\\hline
 \end{tabular}
 \end{center}
\end{example}

We make certain observations on $\bsg$, $\bU_\bsg$ and $\bL_\bsg$.

\newpage

\begin{proposition}\label{prop:obs}
Let $\bsg\in\D^n$. Then the following are true.
\begin{enumerate}[(a)]
\item $\bsg$ is \GC-balanced if and only if $\bU_\bsg$ is balanced.
\item $\bsg'\in \Bindel(\bsg)$ implies that $\bU_{\bsg'}\in \Bindel(\bU_\bsg)$ and $\bL_{\bsg'}\in \Bindel(\bL_\bsg)$.
\item $\bsg'\in \Bedit(\bsg)$ implies that $\bU_{\bsg'}\in \Bedit(\bU_\bsg)$ and $\bL_{\bsg'}\in \Bedit(\bL_\bsg)$.
\end{enumerate}
\end{proposition}

\begin{remark}\label{rem:observation}
The statement in Proposition~\ref{prop:obs} can be made stronger.
Suppose that there is an indel at position $i$ of $\bsg$.
Then there is exactly one indel at the same position $i$ in both upper and lower sequences of $\sigma$. 
For example, consider $\bsg=\tA\tC\tA\tG\tT\tG$ as in Example~\ref{example1}. 
If the third nucleotide $\tA$ is deleted, we obtain $\bsg'=\tA\tC\tG\tT\tG$ and hence,
$\bU'_{\bsg'}=01101$ and $\bL_{\bsg'}=00111$. 
Furthermore, we observe that $\Psi(\bsg)=\bU_\bsg||\bL_\bsg$ suffers a {\em burst of deletions of length two}. 
In our example, $\Psi(\bsg)=0010{\color{red}{00}}110111$ while 
$\Psi(\bsg')=0010110111$.
In Section~\ref{sec:edit}, we make use of this observation to reduce the redundancy of our encoders.
\end{remark}


\subsection{Previous Works}
\label{sec:previous}

The {\em binary VT syndrome} of a binary sequence $\bx\in\{0,1\}^n$ is defined to be
${\rm Syn}(\bx)=\sum_{i=1}^n i x_i$.

For $a \in  \bbZ_{n+1}$, let
\begin{equation}\label{VTcodes}
{\rm VT}_a(n)=\left\{\bx\in \{0,1\}^n: {\rm Syn}(\bx) = a \ppmod{n+1}\right\}.
\end{equation}   
Then ${\rm VT}_a(n)$ form the family of binary codes known as the {\em Varshamov-Tenengolts codes} \cite{le1965}. 
These codes can correct a single indel and Levenshtein later provided a linear-time decoding algorithm \cite{le1965}. 
For any $n$, we know that there exists $a\in\bbZ_{n+1}$ such that ${\rm VT}_a(n)$ has at least $2^n/(n+1)$ codewords.
However, the first known linear-time encoder that maps binary messages into ${\rm VT}_a(n)$ was only described in 1998,
when Abdel-Ghaffar and Ferriera gave a linear-time systematic encoder with redundancy $\ceil{\log (n+1)}$. 

To also correct a substitution, Levenshtein \cite{le1965} constructed the following code 
\begin{equation}\label{VTsub}
{\rm L}_a(n)=\left\{\bx\in \{0,1\}^n: {\rm Syn}(\bx) = a \ppmod{2n}\right\},
\end{equation}  
and provided a decoder that corrects a single edit. However, an efficient encoder that maps binary messages into ${\rm L}_{a}(n)$ was not mentioned in the paper. For completeness, we describe the idea to design such encoder in Subsection~\ref{template}, and we refer it as the {\bf Levenshtein-encoder}, or {\bf Encoder $\mathbb{L}$}.

\begin{theorem}[Levenshtein \cite{le1965}]\label{thm:lev}
Let ${\rm L}_a(n)$ be as defined in \eqref{VTsub}. 
There exists a linear-time decoding algorithm $\dec^L_a:\{0,1\}^{n*}\to {\rm L}_a(n)\cup \{?\}$ such that the following hold.
If $\bc\in {\rm L}_a(n)$ and $\by\in\Bedit(\bc)$, 
then $\dec^L_a(\by)=\bc$.
\end{theorem}

In 1984, Tenengolts \cite{te1984} generalized the binary VT codes to nonbinary ones. 
Tenengolts defined the {\em signature} of a $q$-ary vector $\bx$ of length $n$ to be 
the binary vector $\alpha(\bx)$ of length $n-1$, 
where $\alpha(x)_i=1$ if $x_{i+1}\geq x_i$, and $0$ otherwise, for $i\in[n-1]$.
For $a\in \bbZ_n$ and $b\in \bbZ_q$, set 
  
{
\begin{equation*}
\label{qaryVT}
    {\rm T}_{a,b}({n;q}) \triangleq \left\{ \bx \in \bbZ_q^n : \alpha(\bx)\in{\rm VT}_a(n-1) \text{ and }
    \sum_{i=1}^n x_i = b\ppmod{q} \right\}.
\end{equation*}  
}
  
Then Tenengolts showed that $T_{a,b}(n;q)$ corrects a single indel and 
there exists $a$ and $b$ such that the size of ${\rm T}_{a,b}({n;q})$ is at least $q^n/(qn)$. 
These codes are known to be asymptotically optimal. 

In the same paper, Tenengolts also provided a systematic $q$-ary single-indel-encoder 
with redundancy $\log n +C_q$, where $n$ is the length of a codeword and $C_q$ is independent of $n$.
When $q=4$, we have $7\le C_4\le 10$ for $n\ge 20$.
However, for this encoder, the codewords do not belong to ${\rm T}_{a,b}({n;q})$ for some fixed values of $a$ and $b$.
Recently, Abroshan \et{} provided a method to systematically encode $q$-ary messages into ${\rm T}_{a,b}({n;q})$ \cite{abroshan}. 
However, the redundancy of such encoder is much larger compared with Tenengolts' work. 
Specifically, the encoder of Abroshan \et{} \cite{abroshan} uses $(\log q+1)\lceil\log n\rceil + 2(\log q-1)$ bits of redundancy and particularly, when $q=4$, the redundancy is $3\lceil\log n\rceil+2$.


\subsection{Levenshtein Encoder}\label{template}

Recall the definition of ${\rm L}_a(n)$ in \eqref{VTsub}
and recall that ${\rm L}_a(n)$ is a binary code that can correct a single edit.
We now provide a linear-time encoder that maps binary messages into ${\rm L}_a(n)$. 

\vspace{1mm}

\noindent{\bf Encoder $\mathbb{L}$}. Given $n$, set $t\triangleq\ceil{\log n}$ and $m\triangleq n-t-1$. 

{\sc Input}: $\bx\in \{0,1\}^m$\\
{\sc Output}: $\bc \triangleq \enc_{\mathbb{L}}(\bx)\in {\rm L}_a(n)$\\[-2mm]
\begin{enumerate}[(I)]
\item Set $S \triangleq \{2^{j-1} : j \in [t]\}\cup \{n\}$ and $I \triangleq [n]\setminus S$.
\item Consider $\bc'\in \{0,1\}^n$, where $\bc'|_I=\bx$ and $\bc'|_S=0$. 
Compute the difference $d'\triangleq a-{\rm Syn}(\bc') \ppmod{2n}$. 
In the next step, we modify $\bc'$ to obtain a codeword $\bc$ with ${\rm Syn}(\bc)=a \ppmod{2n}$.
\item We have the following two cases. 
\begin{itemize}
\item \underline{Suppose that $d'< n$}. 
Let $y_{t-1}\ldots y_1y_0$ be the binary representation of $d'$.
In other words, $d' = \sum_{i=0}^{t-1} y_i 2^i$. 
Then we set $c_{2^{j-1}}=y_{j-1}$ for $j\in [t]$, $c_n=0$ and $\bc|_I=\bc'|_I$ to obtain $\bc$. 
\item \underline{Suppose that $n \leq d' < 2n$}. 
We now compute the difference $d''\triangleq d'-n \ppmod{2n}$ and hence, $d'' <n$. 
Consequently, the binary representation of $d''$ is of length $t=\lceil\log n\rceil$ and 
let it be $y_{t-1}\ldots y_1y_0$. 
As before, we set $c_{2^{j-1}}=y_{j-1}$ for $j\in [t]$, $c_n=1$ and $\bc|_I=\bc'|_I$ to obtain $\bc$.
\end{itemize}
\end{enumerate} 

We illustrate Encoder $\mathbb{L}$ via an example.

\begin{example}
Consider $n=10$ and $a=0$. 
Then $t=4$ and $m=5$.
Suppose that the message is $\bx=11011$ and we compute $\enc_{\mathbb{L}}(\bx)\triangleq \bc={\color{blue}{c_1}}{\color{blue}{c_2}}c_3{\color{blue}{c_4}}c_5c_6c_7{\color{blue}{c_8}}c_9{\color{blue}c_{10}} \in {\rm L}_0(10)$. 

\begin{enumerate}[(I)]
\item Set $S=\{1,2,4,8,10\}$ and $I=\{3,5,6,7,9\}$.
\item Then we set $\bc'|_I=11011$ to obtain $\bc'={\color{blue}00}{1\color{blue}0}101{\color{blue}0}1{\color{blue}0}$. We then compute $d'=a- {\rm Syn}(\bc')= 16 \ppmod{20}$. 
\item Since $d'>10$, we compute $d''=d'-10= 6 \ppmod{20}$. The binary representation of $6$ is $0110$. 
Therefore, we set $c_1=0$, $c_2=1$, $c_4=1$, $c_8=0$. Since $d'>10$, we set $c_{10}=1$.
In summary, $\bc={\color{blue}01}{1\color{blue}1}101{\color{blue}0}1{\color{blue}1}$. We can verify that ${\rm Syn}(\bc)=10  \ppmod{20}$.
\end{enumerate}
\end{example}

\begin{theorem} Encoder $\mathbb{L}$ is correct and has redundancy $\ceil{\log n}+1$ bits. 
In other words, $\enc_{\mathbb{L}}(\bx)\in {\rm L}_a(n)$ for all $\bx\in\{0,1\}^m$.
\end{theorem}

\begin{proof}
It suffices to show that ${\rm Syn}(\bc)=a$.
Now, since the words $\bc'$ and $\bc$ differ at the indices corresponding to $S$,
we have that ${\rm Syn}(\bc)={\rm Syn}(\bc') +{\rm Syn}(\bc|_S)$.
When $d'<n$, ${\rm Syn}(\bc|_S)=d'$ and so, ${\rm Syn}(\bc)=a\ppmod{2n}$.
When $d'\ge n$, ${\rm Syn}(\bc|_S)=d''+n=d'$ and again, we have ${\rm Syn}(\bc)=a\ppmod{2n}$, as desired.
\end{proof}
For completeness, we state the corresponding decoder for binary code that correct a single edit.

\vspace{1mm}

\noindent{\bf Decoder $\mathbb{L}$}. For any $n$, set $m=n-\ceil{\log n}-1$.

{\sc Input}: $\by\in \{0,1\}^{n*}$\\
{\sc Output}: $\bx=\dec_{\mathbb{L}}(\by) \in\{0,1\}^m$\\[-2mm]
\begin{enumerate}[(I)]
\item Using Theorem~\ref{thm:lev}, set $\bc\triangleq \dec^L_a(\by)$.
\item Set $\bx\triangleq\bc|_I$, where $I$ is defined by Encoder $\mathbb{L}$.
\end{enumerate} 

It is not hard to use Encoder $\mathbb{L}$ to construct an efficient encoder for DNA alphabet and the output codewords can correct a single edit. For any DNA strand $\sigma$, we can use Encoder $\mathbb{L}$ to encode the upper sequence $\bU_{\bsg}$ and lower sequence $\bL_{\sigma}$ into into ${\rm L}_a(n)$. If $\bU_{\bsg}$ and $\bL_{\sigma}$ can correct a single edit, according to Proposition~\ref{prop:obs}, $\sigma$ can correct a single edit. This construction costs $2 \ceil{\log n} +2 $ bits of redundancy. 

To correct a single indel, we can modify Encoder $\mathbb{L}$ to lower the redundancy to $\ceil{\log n} + 2$ bits when $q=4$.


\section{Encoders Correcting a Single Indel}
\label{sec:indel}

\subsection{Code Construction}
Recall that in Remark~\ref{rem:observation}, we observed that when an indel occurs in $\bsg\in\D^n$,
the binary sequence $\Psi(\bsg)$ has a burst of indels of length two.
In other words, we are interested in binary codes that correct a single burst of indels of length two.
To do so, we have the following construction by Levenshtein \cite{le1965}.

\begin{definition}
For $\bx\in \{0,1\}^n$, we write $\bx$ as the concatenation of $s$ substrings $\bx=\bu_0\bu_1\ldots\bu_{s-1}$, 
where each substring $\bu_i$ contains identical bits,
while substrings $\bu_i$ and $\bu_{i+1}$ contain different bits. 
Each substring $\bu_i$ is also known as a {\em run} in $\bx$. 
Let $r_i$ be the length of the run $\bu_i$. The {\em run-syndrome} of the binary word $\bx$, denoted by ${\rm Rsyn}(\bx)$, is defined as follows.
\begin{equation}\label{rundef}
{\rm Rsyn}(\bx) = \sum_{i=1}^{s-1} i r_i.
\end{equation}
\end{definition}

\begin{example}\label{ex:rsyn-1}
The word $0010110$ has five runs, namely, $\bu_0=00$, $\bu_1=1$, $\bu_2=0$, $\bu_3=11$ and $\bu_4=0$. 
Hence, $r_1=r_2=r_4=1$, $r_0=r_3=2$ and ${\rm Rsyn}(\bx)=13$.
\end{example}

\begin{theorem}[Levenshtein \cite{le1967}]\label{2del}
For $a\in \bbZ_{2n}$,  set
\begin{equation}\label{Bcode}
\Lburst_a(n)\triangleq \{\bx \in \{0,1\}^n: {\rm Rsyn}(0\bx) = a \ppmod{2n} \}.
\end{equation}
Then code $\Lburst_a(n)$ can correct a burst of indel of length two.
Furthermore, there exists a linear-time algorithm $\dec^{\rm burst}_a$ such that 
$\dec^{\rm burst}_a(\by)=\bc$ for all $\by\in \B^{\rm burst}(\bc)$ and $\bc\in\Lburst_a(n)$.
Here, $\B^{\rm burst}(\bc)$ refers to the error ball with respect to a single burst of indel of length two
centered at $\bc$.
\end{theorem}

Using this family of codes, we have a code over $\D$ that corrects a single indel.
For $a\in \bbZ_{2n}$, set
\begin{equation}\label{eq:single-indel-D}
\C_a(n) \triangleq \left\{ \Psi^{-1}(\bc) : \bc \in \Lburst_a(2n) \right\}.
\end{equation}

To design a linear-time encoder for $\C_a(n)$, we make use of the following relation
between run-syndrome and VT-syndrome of a binary word. 

\begin{lemma}[Levenshtein \cite{le1967}]\label{2relation}
Define $\Phi:\{0,1\}^n \to\{0,1\}^n$ such that 
\begin{equation}\label{def:Phi}
\Phi(\bx)_i = 
\begin{cases}
x_i+x_{i+1} \ppmod{2} & \mbox{when }i\in[n-1],\\
x_n & \mbox{when }i=n.
\end{cases}
\end{equation}
Then $\Phi$ is an one-to-one map, and we have that
\begin{equation}\label{relation}
{\rm Rsyn}(0\bx) = -{\rm Syn}(\Phi({\bx})) \ppmod{2n}.
\end{equation}
\end{lemma}

\begin{example}[Example~\ref{ex:rsyn-1} continued]
Consider $\bx=010110$. We have $\Phi(\bx)=111010$ and $-{\rm Syn}(\Phi(\bx))=-11=1\ppmod{12}$. 
On the other hand, $0\bx=0010110$ and indeed, ${\rm Rsyn}(0\bx)=13=1\ppmod{12}$. 
\end{example}

\subsection{Efficient Encoder Correcting a Single Indel}
We now present an efficient method to translate binary sequences into $\C_a(n)$
and hence, obtain a linear-time single-indel-correcting encoder over $\D$. We refer this as Encoder $\mathbb{I}$.

\vspace{1mm}

\noindent{\bf Encoder $\mathbb{I}$}. Given $n$, set $m=2n-\ceil{\log n}-2$.

{\sc Input}: $\bx\in \{0,1\}^m$\\
{\sc Output}: $\bsg = \enc_{\mathbb{I}}(\bx)\in \C_a(n)$, where $\C_a(n)$ is defined in \eqref{eq:single-indel-D}\\[-3mm]
\begin{enumerate}[(I)]
\item Observe that $m=2n-\ceil{\log 2n}-1$. Using Encoder $\mathbb{L}$, compute $\bc\in {\rm L}_{-a}(2n)$.
In other words, ${\rm Syn}(\bc)=-a \ppmod{4n}$. 
\item Compute $\bc'\triangleq\Phi^{-1}(\bc)$ as defined in \eqref{def:Phi}. Hence, ${\rm Rsyn}(0\bc')=a \ppmod{4n}$.
\item Set $\bsg \triangleq \Psi^{-1}(\bc')$.
\end{enumerate} 

\begin{example}
Consider $n=5$, $m=2n- \lceil\log{2n}\rceil -1=5$, $a=0$. We encode $\bx=11000$. 
\begin{enumerate}[(I)]
\item Encode $\bx$ to a codeword $\bc \in {\rm L}_{0}(10)$ using Encoder $\mathbb{L}$. Hence, $\bc=0110100001$. 
\item Next, we compute $\bc'=\Phi^{-1}(\bc)= 001 001 1111$. 
\item Hence, we obtain $\bsg = \Psi^{-1}(\bc')=\mathtt{ACTGG}$.
\end{enumerate}
\end{example}

\begin{theorem} Encoder $\mathbb{I}$ is correct and has redundancy $\ceil{\log n}+2$ bits. 
In other words, $\enc_{\mathbb{I}}(\bx)\in {\C}_a(n)$ for all $\bx\in\{0,1\}^m$.
\end{theorem}

\begin{proof}
Let $\bsg\triangleq \enc_{\mathbb{I}}(\bx)$. 
From Remark~\ref{rem:observation}, it suffices to show that $\bc'\triangleq\Psi(\bsg)$ belongs to $\Lburst_a(2n)$, or  ${\rm Rsyn}(0\bc')=a \ppmod{4n}$. 
This follows from Lemma~\ref{2relation} and the fact that $\bc=\Phi(\bc')$ and ${\rm Syn}(\bc)=-a \ppmod{4n}$.
\end{proof}

\begin{remark} Encoder $\mathbb{I}$ runs in linear-time and the redundancy is $\ceil{\log n}+2$ bits. As mentioned earlier, one may use the systematic $q$-ary single-indel-encoder introduced by Tenengolts \cite{te1984} for $q=4$. The redundancy of such encoder is $\log n +c$, where $7\le c \le 10$ for $n\ge 20$. In other words, in general, Encoder $\mathbb{I}$ improves the redundancy by at least four bits.

\end{remark}
For completeness, we state the corresponding decoder, Decoder $\mathbb{I}$, for DNA codes that correct a single indel.

\vspace{1mm}

\noindent{\bf Decoder $\mathbb{I}$}. For any $n$, set $m=2n-\ceil{\log n}-2$.

{\sc Input}: $\bsg\in \D^{n*}$\\
{\sc Output}: $\bx=\dec_{\mathbb{I}}(\bsg) \in\{0,1\}^m$\\[-2mm]
\begin{enumerate}[(I)]
\item Compute $\widehat{\bc'} \triangleq\Psi(\bsg)$.
\item Using Theorem~\ref{2del}, compute $\bc'\triangleq \dec^{\rm burst}_a\left(\widehat{\bc'}\right)$.
\item Set $\bc\triangleq \Phi(\bc')$.
\item Set $\bx\triangleq \bc|_I$, where $I$ is defined by Encoder $\mathbb{L}$.
\end{enumerate} 

Encoder $\mathbb{I}$ can be extended to obtain linear-time $q$-ary single-indel-correcting encoders with redundancy 
$(1/2 \log q)\log n+O(1)$. 
This improves the encoder of Abroshan \et{} that uses $(\log q+1) \log n + O(1)$ bits of redundancy \cite{abroshan}. 
Unfortunately, unlike Tenegolts' encoder \cite{te1984}, this $q$-ary encoder is not order-optimal.

\subsection{Efficient Encoder for Codes Correcting a Burst of Indels}
Recently, Schoeny \et{}  constructed binary codes that corrects a single burst of indels of length $b$
with $\log n + o(\log n)$ bits of redundancy for fixed values of $b$ \cite{Schoeny:2017}.
Here, we extend our techniques to provide linear-time encoders for the codes of Schoeny \et,
and hence obtain order-optimal linear-time burst-indel-correcting encoders for alphabet of size $q$, $q > 2$. 
In this paper, we focus on the case $q=4$. The work can be easily extended and generalized to any alphabet size. We first introduce the definition of burst of indels. 

\newpage

Let $\bx \in \Sigma^n$. We refer to a {\em $b$-burst of deletions} when exactly $b$ consecutive deletions have occurred, i.e., from $\bx$, and we obtain
a subsequence $\bx'=(x_1,x_2,\ldots,x_i,x_{i+b+1},\ldots,x_n) \in \Sigma^{n-b}$. Similarly, we refer to a {\em $b$-burst of insertions} when exactly $b$ consecutive insertions have occurred, i.e., from $\bx$, and we obtain $\bx''= (x_1,x_2,\ldots,x_j,{\color{red}{y_1,y_2,\ldots,y_b}}, x_{j+1},\ldots,x_n) \in \Sigma^{n+b}$. A {\em $b$-burst of indels} refers to either a $b$-burst of deletions or a $b$-burst of insertions has occurred. We define the {\em b-burst error ball}:
\begin{align*}
\Bindel_{\bburst}(\bx) \triangleq \{\bx\}\cup\{ \by : \by \mbox{ is obtained from $\bx$ via a $b$-burst of indels}\}.
 \end{align*}
 Let $\C\subseteq \Sigma^n$. We say that $\C$ is a {\em $b$-burst-indels-correcting code} if 
$\Bindel_{\bburst}(\bx)\cap \Bindel_{\bburst}(\by)=\varnothing$ for all distinct $\bx,\by \in \C$.

In the binary case, Schoeny \et{} represent the codewords of length $n$ in the $b$-burst-indels-correcting code as $b \times n/b$ codeword arrays,
 where $b$ divides $n$. Thus, for a codeword $\bx$, the codeword array $A_b(\bx)$ is formed by $b$ rows and $n/b$ columns, and the codeword is transmitted column-by-column. Observe that a $b$-burst deletes (or inserts) in $\bx$ exactly one bit 
 from each row of the array $A_b(\bx)$. Here, the $i$th row of the array is denoted by $A_b(\bx)_{\bf {\rm i}}$.
\begin{equation*}
A_b(\bx)=
\left[
  \begin{array}{cccccc}
    x_1 & x_{b+1} & \cdots & x_{(j-1)b+1} & \cdots & x_{(n/b-1)b+1}\\
    x_2 & x_{b+2} & \cdots & x_{(j-1)b+2} & \cdots & x_{(n/b-1)b+2}\\
    \vdots & \vdots & \ddots & \vdots & \ddots & \vdots \\
      x_b & x_{2b} & \cdots & x_{jb} & \cdots & x_{n}\\
  \end{array}
\right].
\end{equation*}

\begin{construction}[Schoeny \et{} \cite{Schoeny:2017}]\label{schoenyconst} Given $n>0,b=o(n)$, Schoeny \et{} then construct a $b$-burst-indels-correcting code of length $n$ as follows. For a codeword $\bx$, the codeword array $A_b(\bx)$ satisfies the following constraints.
\begin{itemize}
\item The first row in the array is a VT-code in which also restricts the longest run of $0$'s or $1$'s to be at most $r=\log 2n+O(1)$. 
Schoeny \et{} also provided the {\em runlength-limited encoder} which uses only one redundancy bit in order to encode binary vectors of maximum run length at most $\ceil{\log n} + 3$ (see \cite[Appendix B]{Schoeny:2017}). 

\item Each of the remaining $(b - 1)$ rows in the array is then encoded using a modified version of the VT-code, 
which they refer as {\em shifted VT (SVT) code}. 
This code corrects a single indel in each row provided the indel position is known to be within $P$ consecutive positions. 
To obtain the desired redundancy, Schoeny \et{} set $P=r+1=\log 2n+O(1)$.
\end{itemize}
\end{construction}

Formally, the following results were provided by Schoeny \et{} \cite{Schoeny:2017}.

\begin{theorem}[Schoeny \et{} \cite{Schoeny:2017}]\label{encoderRLL}
There exists a pair of linear-time algorithms $\enc_{RLL}:\{0,1\}^{n-1}\to \{0,1\}^n$ and 
$\dec_{RLL}:\{0,1\}^n\to \{0,1\}^{n-1}$ such that the following holds.
If $\enc_{RLL}(\bx)=\by$, then $\by$ is a binary vector of maximum run length at most $\ceil{\log n} + 3$ and $\dec_{RLL}(\by)=\bx$.
\end{theorem}

\begin{definition}[Schoeny \et{} \cite{Schoeny:2017}] For $0\leq c<P$ and $d \in \{0,1\}$, the {\em shifted VT-code ${\rm SVT}_{c,d,P}(n)$} is defined as
\begin{align*}
{\rm SVT}_{c,d,P}(n) \triangleq \{\bx: {\rm Syn}(\bx) = c \ppmod{P} \text{ and } \sum_{i=1}^{n} x_i = d \ppmod{2} \}.
\end{align*}

\end{definition}

\begin{theorem}[Schoeny \et{} \cite{Schoeny:2017}]\label{thm:svt}
The code ${\rm SVT}_{c,d,P}(n)$ can correct a single indel given knowledge of the location of the deleted (or inserted) bit to within $P$ consecutive positions. 
Furthermore, there exists a linear-time algorithm $\dec_{c,d,P}^{SVT}$ such that the following holds. 
If $\bc\in{\rm SVT}_{c,d,P}(n)$, $\by\in\Bindel(\bc)$ and the deleted (or inserted) index belongs to $\bbJ\triangleq[j,j+P-1]$
for some $1\le j\le n-P$, then $\dec_{c,d,P}^{SVT}(\by,\bbJ)=\bc$. 
\end{theorem}

We now modify the construction of Schoeny \et{} to obtain a quaternary linear-time $b$-burst-indels-correcting encoder with redundancy $\log n + o(\log n)$. 
Observe that if we convert a quaternary sequence of length $n$ into binary sequence of length $2n$, 
then a $b$-burst-indels in quaternary sequence results in a $2b$-burst-indels in the corresponding binary sequence. 
Suppose that we want to encode messages into quaternary code of length $n=bN$. 

\begin{construction}\label{modified}
Let $n=bN, P,r>0, P\geq r+1$. Given $a\in \mathbb{Z}_{N}, c \in \mathbb{Z}_{P},$ and $d\in \mathbb{Z}_{2}$, let $\C_{\bburst}^{indel}(n,P,r;a,c,d)$ be the code contains all codewords $\bsg \in \D^{bN}$ such that when we view $\bx= \Psi(\bsg) \in \{0,1\}^{2bN}$ as the array $A_{2b}(\bx)$, the following constraints are satisfied.
\begin{itemize}
\item The first row $A_{2b}(\bx)_1 \in {\rm L}_a(N)$ and the longest run of $0$'s or $1$'s is at most $r$.
\item For $2\leq i\leq 2b$, the $i$th row $A_{2b}(\bx)_i \in {\rm SVT}_{c,d,P}(N)$.
\end{itemize}
\end{construction}
Clearly, $\C_{\bburst}^{indel}(n,P,r;a,c,d)$ is a $b$-burst-indels-correcting code. 
Before we provide the encoder for $\C_{\bburst}^{indel}(n,P,r;a,c,d)$, 
an encoder for ${\rm SVT}_{c,d,P}(n)$ is needed. 
One can easily modify the encoder for VT-code to obtain an efficient encoder for shifted VT-code
with redundancy $\ceil{\log P}+1$. For completeness, we describe the encoder for shifted VT-code ${\rm SVT}_{c,d,P}(n)$. 

%
%

\newpage

\noindent{\bf SVT-Encoder}. Given $n,c,d,P$, set $t\triangleq\ceil{\log P}$ and $m\triangleq n-t-1$. 

{\sc Input}: $\bx\in \{0,1\}^m$\\
{\sc Output}: $\bc \triangleq \enc_{SVT}(\bx)\in {\rm SVT}_{c,d,P}(n)$\\[-2mm]
\begin{enumerate}[(I)]
\item Set $S \triangleq \{2^{j-1} : j \in [t]\}\cup \{P\}$ and $I \triangleq [n]\setminus S$.
\item Consider $\bc'\in \{0,1\}^n$, where $\bc'|_I=\bx$ and $\bc'|_S=0$. 
Compute the difference $d'\triangleq a-{\rm Syn}(\bc') \ppmod{P}$. 
In the next step, we modify $\bc'$ to obtain a codeword $\bc$ with ${\rm Syn}(\bc)=a \ppmod{P}$.
\item Let $y_{t-1}\ldots y_1y_0$ be the binary representation of $d'$.
In other words, $d' = \sum_{i=0}^{t-1} y_i 2^i$. 
Then we set $c_{2^{j-1}}=y_{j-1}$ for $j\in [t]$.
\item Finally, we set the bit $x_P$ as the parity check bit that satisfies $x_P=d-\sum_{i\in [n]\setminus {P}} x_i\ppmod{2}$.
\end{enumerate} 

To conclude this subsection, we provide a linear-time encoder for $\C_{\bburst}^{indel}(n;a,c,d)$.

\vspace{1mm}

\noindent{\bf b-Burst-Indels-Encoder}. Given $n=bN,r=2\ceil{\log N}+4, P= r+1, a \in \mathbb{Z}_{N}, c \in \mathbb{Z}_{P}, d\in \mathbb{Z}_{2}$, set $t\triangleq\ceil{\log P}$ and $m\triangleq 2bN-\ceil{\log N}-(2b-1)(t+1)-2$.

{\sc Input}: $\bx\in \{0,1\}^m$\\
{\sc Output}: $\bc \triangleq \enc_{b-burst}^{indel}(\bx)\in \C_{b-burst}^{indel}(n;a,c,d)$\\[-2mm]
\begin{enumerate}[(I)]
\item Set $\bx_1$ be the first $(N-\ceil{\log N}-2)$ bits in $\bx$ and for $2\leq i\leq 2b$, set $\bx_i$ be the subsequence of length $(N-t-1)$ of $\bx$ such that $\bx=\bx_1 \bx_2 \ldots \bx_{2b}$.
\item Let $\by_1'=\enc_{RLL}(\bx_1)$ and use Encoder $\mathbb{L}$ as described in Subsection~\ref{template} to encode $\by_1= \enc_{\mathbb{L}}(\by_1')$.
\item For $2\leq i\leq 2b$, use SVT Encoder $\enc_{SVT}$ to encode $\by_i= \enc_{SVT}(\bx_i) \in {\rm SVT}_{c,d,P}(n)$. 
\item Finally, set $\by=\by_1||\by_2||\ldots||\by_{2b}$ and output $\bsg\triangleq \Psi^{-1}(\by)$.
\end{enumerate} 
For given $n=bN$, the $b$-Burst-Indels-Encoder costs $\ceil{\log N}+O(\log \log N) = \log n + O(\log \log n)$ bits of redundancy. 

\begin{theorem} Let $b=\Theta(1)$. The $b$-Burst-Indels-Encoder is correct. In other words, $\enc_{b-burst}^{indel}(\bx)\in \C_{b-burst}^{indel}(n,P,r;a,c,d)$ for all $\bx\in\{0,1\}^m$.
\end{theorem}

\begin{proof}
Let $\bsg\triangleq \enc_{b-burst}^{indel}(\bx)$ and $\by= \Psi(\bsg)$. It is sufficient to show that when we view $\by$ as the array $A_{2b}(\by)$ the constraints in Construction~\ref{modified} are satisfied. Based on our encoder, $\by=\by_1||\by_2||\ldots||\by_{2b}$ and $\by_i= \enc_{SVT}(\bx_i) \in {\rm SVT}_{c,d,P}(n)$ for $2\leq i\leq 2b$. It remains to show that the first row $\by_1 \in {\rm L}_a(N)$ and the longest run of $0$'s or $1$'s is at most $r$. Observe that $\by_1'=\enc_{RLL}(\bx_1)$, which implies the longest run of $0$'s or $1$'s in $y_1'$ is at most $(\ceil{\log N}+3)$ according to Theorem~\ref{encoderRLL}. Therefore, the maximum run in $\by_1$ after $\enc_{\mathbb{L}}$ is at most $(\ceil{\log N}+3)+(\ceil{\log N}+1)=2\ceil{\log n}+4=r$ (refer to Subsection III-C, Encoder $\mathbb{L}$).
\end{proof}





\section{Encoders Correcting a Single Edit}\label{sec:edit}
In this section, we present efficient encoders for codes over $\D^n$ that can correct a single edit. Unlike indel error, the main challenge to design codes correcting edit is that when substitution error occurs in a DNA strand $\bsg$, it might not affect $\bU_\bsg$ or $\bL_\bsg$. Therefore, putting a VT constraint in either one of these sequences might not tell any information about the loss in the other sequence. We illustrate this scenario through the example below.

\begin{example}\label{editexample}
Suppose that $\bsg=\tA\tC\tA\tG\tT\tG$. Suppose a substitution error occurs at the third symbol, changing $\tA$ to $\tT$, and we received $\bsg_1=\tA\tC{\color{red}{\tT}}\tG\tT\tG$.  
On the other hand, suppose a substitution error also occurs at the third symbol, changing $\tA$ to $\tG$, and we received $\bsg_2=\tA\tC{\color{red}{\tG}}\tG\tT\tG$. We then see the change in $\bU_\bsg$ and $\bL_\bsg$.

\begin{center}
\begin{tabular}{ll}
\begin{tabular}{|c|| c | c| c| c| c|c|}
 \hline
 $\bsg$ & $\tA$ & $\tC$ & $\tA$  & $\tG$ & $\tT$ & $\tG$ \\ \hline
 $\bU_{\bsg}$ & 0 & 1 & 0 & 1 & 0 & 1 \\\hline
 $\bL_{\bsg}$ & 0 & 0 & 0 & 1 & 1 & 1\\\hline
 \end{tabular}

\begin{tabular}{|c|| c | c| c| c| c|c|}
 \hline
 $\bsg_1$ & $\tA$ & $\tC$ & ${\color{red}{\tT}}$  & $\tG$ & $\tT$ & $\tG$ \\ \hline
 $\bU_{\bsg}$ & 0 & 1 & 0 & 1 & 0 & 1 \\\hline
 $\bL_{\bsg}$ & 0 & 0 & {\color{red}{1}} & 1 & 1 & 1\\\hline
 \end{tabular}

\begin{tabular}{|c|| c | c| c| c| c|c|}
 \hline
 $\bsg_2$ & $\tA$ & $\tC$ & ${\color{red}{\tC}}$  & $\tG$ & $\tT$ & $\tG$ \\ \hline
 $\bU_{\bsg}$ & 0 & 1 & {\color{red}{1}} & 1 & 0 & 1 \\\hline
 $\bL_{\bsg}$ & 0 & 0 & 0 & 1 & 1 & 1\\\hline
 \end{tabular}
\end{tabular}
\end{center}

\end{example}

According to Proposition~\ref{prop:obs}, if $\bsg'\in \Bedit(\bsg)$ implies that $\bU_{\bsg'}\in \Bedit(\bU_\bsg)$ and $\bL_{\bsg'}\in \Bedit(\bL_\bsg)$. Therefore, a simple solution is to encode both of the sequences $\bU_\bsg$ and $\bL_\bsg$ into ${\rm L}_a(n)$ using Encoder $\mathbb{L}$. Recall that ${\rm L}_a(n)$ can detect and correct a single edit. Hence,  $\bsg=\Psi^{-1}(\bU_\bsg || \bL_\bsg)$ can correct a single edit. This simple encoder costs $2\ceil{\log n}+2$ bits of redundancy, which is at most twice the optimal. 
For completeness, we present the encoder below and refer this as the {\em Encoder A}.
We remark that Encoder A is crucial in construction of {\GC}-balanced codebooks in Section~\ref{sec:balanced}.

\subsection{First Class of Single-Edit-Correcting Codes}

\begin{theorem}\label{thm:single-edit}
Set $m=2(n-\ceil{\log n}-1)$.
There exists a linear-time single-edit-correcting encoder $\enc_{\mathbb{E}}^{A}:\{0,1\}^m \to \D^n$ 
with redundancy $2\ceil{\log n}+2$.
\end{theorem}

\begin{proof}We describe the single-edit-correcting encoder.
Consider the message $\bx_1\bx_2\in \{0,1\}^m$ with $|\bx_1|=|\bx_2|=n-\ceil{\log n} -1$.
For $i\in [2]$, set $\bc_i=\enc_{\mathbb{L}}(\bx_i)\in \{0,1\}^n$. Then set $\enc_{\mathbb{E}}^{A}(\bx_1\bx_2)=\Psi^{-1}(\bc_1||\bc_2)$.
\end{proof}

\begin{remark}
An alternative approach to correct a single edit is to consider the quaternary Hamming Code $\C_H$
with $\log (3n+1)\approx \log n + 1.58$ bits of redundancy.
If we partition the codewords in $\C_H$ into $4n$ equivalence classes according to their VT parameters, 
we are then guaranteed a code that corrects a single edit with redundancy at most $2\log n+ 3.58$.
In contrast, the linear-time encoder in Theorem~\ref{thm:single-edit} has redundancy $2\ceil{\log n}+2$.

Furthermore, the best known linear-time encoder that maps messages into one such class is the one by Abroshan \et{} 
and the encoder introduces additional $3\ceil{\log n}+2$ redundant bits \cite{abroshan}. 
Thus, an efficient single-edit-correcting encoder obtained via the construction has redundancy approximately $4\log n+3.58$.
\end{remark}

\subsection{Order-Optimal Quaternary Codes Correcting a Single Edit}

In this subsection, we consider the quaternary alphabet $\Sigma_4=\{0,1,2,3\}$ as a subset of integers.
Our main result in this subsection is the existence of a quaternary single-edit-correcting code 
that has redundancy $\log n + O(\log \log n)$, where $n$ is the length of the codeword. 

Crucial to our construction is the {\em sum-balanced constraint}, which is defined as below. 

\begin{definition}\label{sumbalance} Let $\bx=(x_1,x_2,\ldots,x_n) \in \Sigma_4^n$. A window $\bW$ of length $k$ of $\bx$, i.e. $\bW=(x_{i+1},\ldots,x_{i+k})$ is called {\em sum-balanced} if $k<  \sum_{x_j\in W} x_j <2k$. 
A word $\bx$ is {\em $k$-sum-balanced} if every window $\bW$ of the word $\bx$ is sum-balanced whenever the window length is at least $k$.
\end{definition}

A key ingredient of our code construction is the set of all $k$-sum-balanced words. 
 \[{\rm Bal}_k(n)\triangleq \left\{ \bx\in \Sigma_4^n: \bx {\text{ is $k$-sum-balanced}}\right\}.\]

We have the following properties of ${\rm Bal}_k(n)$.
The first lemma states that whenever $k=\Omega(\log n)$, the set ${\rm Bal}_k(n)$ incurs at most one symbol of redundancy.

\begin{lemma} Given $n\ge 4$, if $k=36\log n$, then $|{\rm Bal}_k(n)| \geq 4^{n-1}$.
\end{lemma}

To prove this lemma, we require {\em Hoeffding's inequality} \cite{Hoeffding.1963}.

\begin{theorem}[Hoeffding's Inequality]
Let $Z_1, Z_2,\ldots, Z_n$ be independent bounded random variables such that $a_i\le Z_i\le b_i$ for all $i$.
Let $S_n=\sum_{i=1}^n Z_i$. For any $t>0$, we have
\begin{equation*}
P(S_n-E[S_n]\geq t) \leq e^{-{2t^2}/{\sum_{i=1}^n (b_i-a_i)^2}}.
\end{equation*}
\end{theorem}

\begin{proof}[Proof of Lemma 25]  
Let $\bx$ be a uniformly at random selected element from $\Sigma_4^n$.
We evaluate the probability that the first $k$-length window $\bW$ of $\bx$ 
does not satisfy the sum-balanced constraint.  Applying Hoeffding's inequality we obtain:
\begin{equation*}
P \left( \left|\sum_{x_i\in W} x_i - \frac{3k}{2}\right| \geq \frac{k}{2} \right) = 
2 P \left( \left(\sum_{x_i\in W} x_i - \frac{3k}{2}\right) \geq \frac{k}{2} \right) \leq 2 e^{-\frac{2k^2/4}{9k}}=2e^{-{k}/{18}}.
\end{equation*}
Since $f(k)=e^{-{k}/{18}}$ is decreasing in $k$, 
the probability that a $k'$-length window violates the sum-balanced constraint is at most $f(k)$ for $k'\ge k$.
Also, since there are at most $n^2$ windows, applying the union bound and setting $k=36\log n$ yield
\begin{equation*}
P(\bx \notin B_k(n)) \leq n^2 2e^{-\frac{k}{18}} = 2n^2 e^{-2 \log n} = 2 n^{2-2 \log e}.
\end{equation*}
Therefore, the size of $B_k(n)$ is at least
\begin{equation*}
|B_k(n)| \geq 4^n (1- 2 n^{2-2 \log e}).
\end{equation*}
Since $n\geq 4$, we have that $1- 2 n^{2-2 \log e} \geq 1/4$. Therefore, $|B_k(n)| \geq 4^{n-1}$.  
\end{proof}

Next, we recall that the signature of $\bx$, denoted by $\alpha(x)$, is a binary vector of length $n-1$, 
where $\alpha(x)_i=1$ if $x_{i+1}\geq x_i$, and $0$ otherwise, for $i\in[n-1]$. 
It is well-known that if a single deletion occurs in $\bx$, resulting in $\by$, then $\alpha(\by)$ can be obtained by deleting a single symbol from $\alpha(\bx)$.
The following lemma provides an upper bound on the length of a run of a $k$-balanced word and its signature.

\begin{lemma}\label{maxrun}
Let $\bx \in {\rm Bal}_k(n)$.
Then the length of a run in $\bx$ is at most $(k-1)$ while the length of a run in $\alpha(\bx)$ is strictly less than $2(k-1)$. 
\end{lemma}

\begin{proof}
Let $\bx \in {\rm Bal}_k(n)$.
We first show that the length of a  run in $\bx$ is at most $(k-1)$.
Suppose otherwise that the run in $\bx$ is $(x_{i+1},\ldots,x_{i+t})$ where $t \geq k$. 
Since $\bx \in {\rm Bal}_k(n)$, we have $ x_{i+1}+x_{i+2}+\cdots+x_{i+t} = t x_{i+1} \in (t, 2t)$. 
We have a contradiction since $x_{i+1}\in\{0,1,2,3\}$. 

Let $W_{(0,1)}$ be a window in $\bx$ that contains only $0$ and $1$ symbols. 
We claim that the length of $W_{(0,1)}$ is at most $(k-1)$. 
Otherwise, assume that the size is $s$ where $s\geq k$. 
Then the sum of symbols in this window is at most $s$. 
We then get a contradiction since $\bx \in {\rm Bal}_k(n)$ and the sum of such symbols is strictly more than $s$. Similarly, let $W_{(2,3)}$ be a window in $\bx$ that contains only $2$ and $3$ symbols.
Then the length of $W_{(2,3)}$ is at most $(k-1)$. 

We look at runs in the signature $\alpha(\bx)$.
First, the run of zeroes in $\alpha(\bx)$ is at most three and this happens when the corresponding window is $(3,2,1,0)$.
Next, we look at a run of ones in $\alpha(\bx)$.
Now, such a run is obtained when there is a subsequence $\by$ in $\bx$ of the form $\by=0^{t_1}1^{t_2}2^{t_3}3^{t_4}$. As shown earlier, the window $0^{t_1}1^{t_2}$ has length at most $(k-1)$ and the window $2^{t_3}3^{t_4}$ has length at most $(k-1)$. Hence, the length of $\by$ is at most $(2k-2)$, which the corresponding run of ones in 
$\alpha(\bx)$ is at most $2(k-2)$. 
%
%
\end{proof}
 
We are now ready to present our main code construction for this subsection.

\begin{construction}\label{constr:orderoptimal}
Given $k<n$, set $P=5k$.
For $a\in \mathbb{Z}_{4n+1}$, $b \in \mathbb{Z}_{P}$, $c \in \mathbb{Z}_{2}$ and $d\in \mathbb{Z}_{7}$,
set $\C^B(n;a,b,c,d)$ as follows.
\begin{align*}
\C^B(n;k,a,b,c,d) = \Big \{ \bx \in {\rm Bal}_k(n):\, {\rm Syn}(\bx)= a \ppmod{4n + 1},
\, \alpha(\bx) \in {\rm SVT}_{b,c,P}(n-1), \text{ and }  \sum_{i=1}^n x_i = d \ppmod{7} \Big \}.
\end{align*}
\end{construction}

In what follows, we first prove the correctness of Construction~\ref{constr:orderoptimal} by providing an efficient decoder that can correct a single edit in linear time. 
Subsequently, in Theorem~\ref{thm:orderoptimal}, we show that $\C$ is order-optimal with a suitable choice of $k$.

The decoding operates as follows.
\begin{itemize}
\item First, the decoder decides whether a deletion, insertion or substitution has occurred. 
Note that this information can be recovered by simply observing the length of the received word.
\item If the length of the output word is equal to $n$, then we conclude that at most a single substitution error has occurred and Lemma~\ref{decodesub} provides the procedure to the substitution error (if it exists).
\item Otherwise, the output vector has length $n-1$ (or $n+1$). We then conclude that a single deletion (or insertion) has occurred and we proceed according to Lemma~\ref{decodedel}.
\end{itemize}


\begin{lemma}\label{decodesub}
The code $\C^B(n;a,b,c,d)$ corrects a single substitution in linear time.
\end{lemma}
\begin{proof}
Suppose that the original codeword is $\bx$ and  we receive a vector $\by$ of length $n$. 
The decoder proceeds as follows.
\begin{itemize}
\item {\bf Step 1. Error detection}. Compute $d'\triangleq \sum_{i=1}^n y_i \ppmod{7}$. If $d'=d$, then we conclude that there is no error and $\by$ is the original codeword $\bx$. Otherwise, a substitution error occurs. 
Henceforth, we assume that it occurs in position $j$. 
\item {\bf Step 2.} Next, we compute $y_j-x_j=d'-d \ppmod{7}$. Since, $x_j, y_j\in\{0,1,2,3\}$, we can determine the value of $y_j -x_j$ as integers.
\item {\bf Step 3. Error location}. Compute $a'={\rm Syn}(\by) \ppmod{4n + 1}$ and clearly, we have $a'-a= j (y_j-x_j) \ppmod{4n + 1}$. 
Now, since $1\le j\le n$, we uniquely determine the index $j$ with the value of $y_j-x_j$ from Step 2. 
\item {\bf Step 4. Recovering the symbol}.
Finally, given the error location $j$, we recover $x_j$ using $y_j$ and $y_j-x_j$.
\end{itemize}
It is easy to see that all the decoding steps run in $O(n)$ time. 
Hence, $\C^B(n;a,b,c,d)$ corrects a single substitution in linear time.
\end{proof}

It remains to describe the decoding procedure for correcting a single deletion.
To this end, we have the following technical lemma that that characterizes words
 whose deletion balls intersect and whose syndromes are the same.

\begin{lemma}\label{distance}
Let $\bx$ and $\bz$ be two words such that the following hold.
\begin{enumerate}[{\rm (B1)}]
\item $\bx$ and $\bz$ belongs to ${\rm Bal}_k(n)$.
\item ${\rm Syn}(\bx)={\rm Syn}(\bz) =a \pmod{4n+1}$.
\item $\sum_{i=1}^n x_i= \sum_{i=1}^n z_i =d \pmod{7}$.
\item $\B^{\rm indel}(\bx)\cap\B^{\rm indel}(\bz)$ is non-empty.
\end{enumerate}
Suppose that  $\by\in \B^{\rm indel}(\bx)\cap\B^{\rm indel}(\bz)$.
If $\by$ is obtained from $\bx$ by deleting $x_i$ and obtained from $\bz$ by deleting $z_j$,
then we have that $|i-j|\le k$.
\end{lemma}


\begin{proof}
From (B3) and the fact that $x_i$ and $z_j$ belongs to $\Sigma_4$, we have that the deleted symbols are the same. In other words,  $x_i=z_j$ and we set this value to be $m$.

Without loss of generality, assume that $i>j$. We compute the syndrome of $\by$
\[{\rm Syn}(\by) = \sum_{t=1}^{n-1} t y_t = a' \ppmod{4n+1},\]
and consider the quantity
\begin{equation}\label{sample}
(a-a') \ppmod{4n + 1}  = \sum_{t=1}^{n} t x_t - \sum_{t=1}^{n-1} t y_t = im + \sum_{t=i}^{n-1} y_{t}.
\end{equation}

On the other hand, if we compute the same quantity using $\bz$, we have that 
\begin{equation}\label{sample2}
(a-a') \ppmod{4n + 1}  = jm + \sum_{t=j}^{n-1} y_{t}.
\end{equation}

Now, we know that $y_t=z_{t+1}$ for $t\ge j$. 
Subtracting \eqref{sample2} from \eqref{sample} and setting $i-j=b$,
we have that 
\[bm = \sum_{t=j+1}^i z_t \ppmod{4n+1}.\]

Let $b=i-j$. Since $bm < 4n + 1$ and $0\leq z_t\leq 3$ for all $t$, 
we have that $bm = \sum_{t=j+1}^{i} z_{t}$. 
Suppose to the contrary that $i-j=b \geq k$. 
Then since $\bz\in{\rm Bal}_k(n)$, we have that $b < \sum_{t=j+1}^i z_t < 2b$.
Hence, we have that $b<bm<2b$, contradicting the fact that $m$ is an integer.
Therefore, $i-j\le k$.
\end{proof}

Finally, we present the deletion-correcting procedure.

\begin{lemma}\label{decodedel}
The code $\C^B(n;a,b,c,d)$ corrects a single deletion or single insertion in linear time.
\end{lemma}
\begin{proof}
Since the decoding process for correcting an insertion is similar to correcting a deletion, we only present the case of a deletion.
Now, let $\by$ be the result of a deletion occurring to $\bx \in \C^B(n;a,b,c,d)$ at position $i$. 
To recover $\bx$, the decoder proceeds as follows.
\begin{itemize}
\item {\bf Step 1. Identifying the deleted symbol}. From the constraint $\sum_{t=1}^n x_t \equiv d \ppmod{7}$, we compute the deleted symbol $m\triangleq d-\sum_{t=1}^{n-1} y_t \ppmod{7}.$
\item {\bf Step 2. Localizing the deletion}. 
\begin{itemize}
\item {\bf Localizing the deletion in $\bx$}.
Using \eqref{sample} or \eqref{sample2}, we compute the possible deleted positions.
Specifically, set $a'={\rm Syn}(\by)\ppmod{4n+1}$ and we compute $\mathbb{J}=\{1\leq j\leq n: a'+jm+\sum_{t=j}^{n-1} y_t=a \ppmod{4n+1} \} $. 
According to Lemma~\ref{distance}, we have $|i-j| \leq k$ for all $i,j \in \mathbb{J}$. 
\item {\bf Localizing the deletion in $\alpha(\bx)$}.
For $j\in \bbJ$, set $\by^+_j$ to be word obtained by inserting $m$ at index $j$.
Suppose that $\alpha(\by)$ can be obtained by deleting a single symbol from $\alpha(\by^+_j)$ at position $j'$. Then we add $j'$ to the set $\bbJ'_{j}$. We set $\bbJ' \triangleq \bigcup_{j\in \bbJ} \bbJ'_j$.
We claim that $\bbJ'\subseteq [j'_{min}, j'_{min}+P-1]$ where $j'_{min}$ is the smallest index.

Indeed, Lemma~\ref{maxrun} states that the longest run in $\bx$ is at most $(k-1)$ and the longest run in $\alpha(\bx)$ is less than $2(k-1)$. Therefore, the interval containing $\bbJ'$ is of length at most 
$2(k-1)+k+2(k-1)<5k=P$.
\end{itemize}

\item {\bf Step 3. Recovering $\alpha(\bx)$}. 
Since $\alpha(\bx) \in SVT_{b,c,P}(n-1)$,
we apply $\dec_{c,d,P}^{SVT}$ (from Theorem~\ref{thm:svt}) to $\alpha(\by)$ and $\bbJ'$
 to recover $\alpha(\bx)$.
\item {\bf Step 4. Recovering $\bx$}. 
From recovered $\alpha(\bx)$ and symbol $m$, we can determine $\bx$ uniquely.
\end{itemize}
It is easy to see that all the decoding steps run in $O(n)$. Hence, $\C(n;a,b,c,d)$ can correct a single deletion or single insertion in linear time.
\end{proof}
Combining the results of Lemma~\ref{decodesub} and Lemma~\ref{decodedel} we have the following theorem.

\begin{theorem}\label{thm:orderoptimal}
The code $\C^B(n;a,b,c,d)$ corrects a single edit in linear-time. 
There exist $a,b,c,d$ such that the size of $\C^B(n;a,b,c,d)$ is at least 
\begin{equation*}
|\C(n;a,b,c,d)|\geq \frac{|{\rm Bal}_k(n)|}{35 (4n+1) k }.
\end{equation*}
When $k=36\log_4 n$, we have that the redundancy is at most $\log_4 n+ O(\log \log n)$ bits.
\end{theorem}

\begin{proof}
It remains to demonstrate the property of the code size.
The lower bound is obtained using pigeonhole principle. 
Setting $k=36\log_4 n$, we have that $|{\rm Bal}_k(n)|\ge 4^{n-1}$ and hence the redundancy is $\log_4 n+ O(\log \log n)$.
\end{proof}

\begin{remark}
Observe that $\C^B(n;a,b,c,d)$ requires at least 
$1+\log\left(35(4n+1)k\right)\ge 13 +\log n$ bits of redundancy.
In contrast, the encoder $\enc_{\mathbb{E}}^{A}$ from Theorem~\ref{thm:single-edit} requires at most $4+2\log n$ bits of redundancy.
Therefore, even though $\C^B(n;a,b,c,d)$ is order-optimal, the encoder $\enc_{\mathbb{E}}^{A}$ incurs less redundant bits whenever $n\le 2^9=512$.
\end{remark}

\vspace{2mm}

Finally, we provide an efficient encoder that encodes binary messages into a quaternary codebook that corrects a single edit.
While the codebook in general is not the same as $\C^B(n;a,b,c,d)$, the decoding procedure is similar and 
the number of redundant bits is similar to that of $\C^B(n;a,b,c,d)$.

A high-level description of the encoding procedure is as follows.
\begin{enumerate}[(I)]
\item {\bf Enforcing the $k$-sum-balanced constraint}. 
Given an arbitrary message $\bx$ over $\Sigma_4$ of length $m-1$, we encode $\bx$ to a word $\bz\in {\rm Bal}_k(m)$.
However, it is not straightforward to efficiently code for this constraint. 
Hence, we instead impose a tighter constraint on the sum, but we impose the constraint only on windows of length exactly $k$. 
We provide the formal definition of {\em restricted-sum-balanced} in Definition~\ref{rsumbalance}.
\item {\bf Appending the syndromes}. Using $\bz$, we then compute its VT syndrome $a\triangleq{\rm Syn}(\bz) \ppmod{4n + 1}$, 
 SVT syndrome $b\triangleq {\rm Syn}(\alpha(\bz))\ppmod{P}$ and $c\triangleq \sum_{i=1}^n \alpha(\bz)_i \ppmod{2}$,
 and the check $d\triangleq \sum_{i=1}^n z_i  \ppmod{7}$.
 Finally, we append the quaternary representations of $a$, $b$, $c$ and $d$ to the end of $\bz$ with a {\em marker} between $\bz$ and these syndromes. 
 It turns out we can modify the procedures given in Lemma~\ref{decodesub} and Lemma~\ref{decodedel}
 and correct a single edit in linear time. We do so in Theorem~\ref{thm:encoderB}.
\end{enumerate}

\vspace{2mm}

\noindent{\bf Enforcing the $k$-sum-balanced constraint}. As mentioned earlier, instead of encoding into $k$-sum-balanced words, we require the words to have the following property.

\begin{definition}\label{rsumbalance} Let $\bx=(x_1,x_2,\ldots,x_n) \in \Sigma_4^n$. A window $\bW$ of length $k$ of $\bx$, i.e. $\bW=(x_{i+1},\ldots,x_{i+k})$ is called {\em restricted-sum-balanced} if $5k/4<  \sum_{x_j\in W} x_j <7k/4$. 
A word $\bx$ is {\em $k$-restricted-sum-balanced} if every window $\bW$ of length exactly $k$ of the word $\bx$ is restricted-sum-balanced.
\end{definition}

The following lemma states that a  $k$-restricted-sum-balanced is also  $(4k)$-sum-balanced.

\begin{lemma} 
Let ${\rm Bal}_k^{*}(n)=\{\bx \in \Sigma_4^n: \text{ $\bx$ is $k$-restricted-sum-balanced}\}.$ 
We have that ${\rm Bal}_k^{*}(n) \subseteq {\rm Bal}_{4k}(n)$.
\end{lemma}

\begin{proof}
Let $\bx \in {\rm Bal}_k^{*}(n)$ and $\bW$ be a length of $\bx$ that length at least $4k$. 

Suppose that the length of $\bW$ is $Mk+N$ for some $M\geq4$ and $0\leq N<k$. 
Hence, we write $\bW\triangleq\bz_1\bz_2\ldots \bz_M \bz_{M+1}$ 
where each $\bz_i$ is a window of length $k$ for $1\leq i\leq M$ and 
the window $\bz_{M+1}$ is of length $N$. 

Since $\bx \in B_k^{*}(n)$, we then have $\bz_i$ is $k$-restricted-sum-balanced for $1\leq i\leq M$. Therefore, $5k/4 < \sum_{x_j \in \bz_i} x_j < 7k/4$ for $1\leq i\leq M$. 
Hence, for $M\geq 4$ and $0\leq N<k$, we have
\[
\sum_{x_j \in W} x_j 
> \sum_{i=1}^{M} \left(\sum_{x_t \in \bz_i} x_t\right) 
> 5Mk/4 =Mk+ Mk/4 \geq Mk + k > Mk +N,\]
and
\[
\sum_{x_j \in W} x_j 
< \sum_{i=1}^{M} \left(\sum_{x_t \in \bz_i} x_t\right) + 3N
< 7Mk/4 + 3N = 2Mk+ 2N+(N-Mk/4) < 2Mk+2N.\]
This shows that $\bW$ is sum-balanced. Hence, $\bx \in B_{4k}(n)$.
\end{proof}

Now, we may modify the sequence replacement techniques \cite{Elishco.2019,Wijngaarden.2010} to encode for the restricted-sum-balanced constraint.

\begin{theorem}\label{thm:rsb-encoder}
Suppose that $k\ge 72 \log n$. 
Then there is a pair of efficient algorithms $\enc_{\rsb{k}}:\Sigma_4^{n-1}\to \Sigma_4^{n}$ 
and $\dec_{\rsb{k}}:{\rm Im}(\enc_{\rsb{k}})\to \Sigma_4^{n-1}$
such that $\enc_{\rsb{k}}(\bx)\in {\rm Bal}^*_k(n)$ and $\dec_{\rsb{k}}\circ \enc_{\rsb{k}}(\bx)=\bx$ for all $\bx\in \Sigma_4^{n-1}$.
\end{theorem}


\vspace{2mm}

\noindent{\bf Appending the syndromes}. Set $k=72\log n$ in Theorem~\ref{thm:rsb-encoder}.
Suppose that the message  $\bx \in \Sigma_4^{n-1}$ is encoded to $\by=(y_1,y_2,\ldots,y_n) \in {\rm Bal}_k^{*}(n)$,
or, $\by\triangleq \enc_{\rsb{k}}(\bx)$.

Let $k'=4k$ and $P=5k'$. The encoder computes the following four sequences over $\Sigma_4$.
\begin{itemize}
\item Set the {\em VT syndrome} $a\triangleq{\rm Syn}(\by) \ppmod{4n + 1}$. Let $\bR_1$ be the quaternary representations of $a$ of length $\ceil{\log_4(4n+1)}$. 
\item Set the {\em SVT syndrome} $b\triangleq {\rm Syn}(\alpha(\by))\ppmod{P}$. Let $\bR_2$ be the quaternary representations of $b$ of length $\ceil{\log_4 P}$.
\item Set the {\em parity check} of $\alpha(\by)$, i.e. $c\triangleq \sum_{i=1}^n \alpha(\by)_i \ppmod{2}$. Let $\bR_3 \equiv c$. 
\item Set the {\em check} $d\triangleq \sum_{i=1}^n y_i  \ppmod{7}$. Let $\bR_4$ be the quaternary representations of $d$ of length 2. 
\end{itemize}

 Let $r$ be the smallest symbol in $\Sigma_4\setminus\{y_n\}$ and let the marker $\bM=(r,r)$. 
 We append $\bM,\bR_1,\bR_2,\bR_3,\bR_4$ to $\by$ and output the codeword $\by\bM\bR_1\bR_2\bR_3\bR_4$ of length $N=(n+\ceil{\log_4(4n+1)}+\ceil{\log_4 P}+5)$.

Finally, we summarize our encoding procedure and demonstrate its correctness.

\begin{theorem}\label{thm:encoderB}
Given $n$, $a$, $b$, $c$, $d$, set $k=72\log n$, $k'=4k$, $P=5k'$ and $N=(n+\ceil{\log_4(4n+1)}+\ceil{\log_4 P}+5)$.
We define $\enc_{\mathbb{E}}^{B}:\Sigma_4^{n-1} \to \Sigma_4^N$ as follows.
Set $\by=\enc_{\rsb{k}}(\bx)$ as in Theorem~\ref{thm:rsb-encoder} and let $\bM, \bR_1, \ldots, \bR_4$ be as defined above.  
If we set $\enc_{\mathbb{E}}^{B}(\bx)\triangleq \by\bM\bR_1\bR_2\bR_3\bR_4$,
then the code defined by $\enc_{\mathbb{E}}^{B}$ corrects a single edit in linear time.
Therefore, the redundancy of $\enc_{\mathbb{E}}^{B}$ is $\log n+ O(\log \log n)$ bits.
\end{theorem}

\begin{proof}
It remains to provide the corresponding decoder and show that it corrects a single edit in linear time. 
Suppose that we receive $\bz'$. The idea is to recover $\by$ as the first $n$ symbols in $\bz'$ and then use the decoder in Theorem~\ref{thm:rsb-encoder} to recover the information sequence $\bx$. First, the decoder decides whether a deletion, insertion or substitution has occurred. Note that this information can be recovered by simply observing the length of the received word. The decoding operates as follows.

\begin{enumerate}[(i)]
\item {\bf If the length of $\bz'$ is $N$}, we conclude that at most a single substitution error has occurred. Let $\bz'=(z'_1,\ldots,z'_N)$. The decoder sets $\by'$ as the first $n$ symbols of $\bz'$, the marker $\bM=(z'_{n+1}, z'_{n+2})$, $\bR'_1$ as the next $\ceil{\log_4(4n+1)}$ symbols, $\bR'_2$ as the next $\ceil{\log_4 P}$ symbols, $\bR'_3$ as the next symbol and the last two symbols as $\bR'_4$. The decoder proceeds as follows.

\begin{itemize}
\item {\bf Checking the marker}. 
If $z'_{n+1}, z'_{n+2}$ are not identical, then the substitution occurred here. The decoder conclude that $\by \equiv \by'$. On the other hand, if $z'_{n+1}, z'_{n+2}$ are identical, the decoder concludes that there is no error in the marker. The decoder computes $a'\triangleq{\rm Syn}(\by') \ppmod{4n + 1}$, $b'\triangleq {\rm Syn}(\alpha(\by'))\ppmod{P}$, $c'\triangleq \sum_{i=1}^n \alpha(\by')_i \ppmod{2}$ and $d'\triangleq \sum_{i=1}^n y'_i  \ppmod{7}$, and proceeds to the next step.

\item {\bf Comparing the check}.
 If $\bR'_4$ corresponds the quaternary representation of $d'$, the decoder concludes that $\by \equiv \by'$. Otherwise, it proceeds to the next step. 

\item {\bf Comparing the VT syndrome}.
If $\bR'_1$ corresponds to the quaternary representations of $a'$, the decoder concludes that $\by \equiv \by'$. Otherwise, there must be a substitution in the $\by'$. Hence, there is no error in $\bR'_1, \bR'_2, \bR'_3,$ and $\bR'_4$. The decoder follows the steps in Lemma~\ref{decodesub} to recover $\by$ from $\by'$. 
\end{itemize}

\item {\bf If the length of $\bz'$ is $(N-1)$}, we conclude that a single deletion has occurred. Suppose $\bz'=(z'_1,\ldots,z'_{N-1})$. The decoder proceeds as follows.
\begin{itemize}
\item {\bf Localizing the deletion.} If $z'_n$ and $z'_{n+1}$ are different, the decoder concludes that there is no deletion in $\by$ and sets $\by$ as the first $n$ symbols of $\bz'$. Otherwise, there is a deletion in the first $n$ symbols. Hence, there is no error in $\bR_1, \bR_2, \bR_3,$ and $\bR_4$. 
\item {\bf Recovering $\by$.} The decoder sets $\by'$ as the first $(n-1)$ symbols in $\bz'$ and follows the steps in Lemma~\ref{decodedel} to recover $\by$ from $\by'$.
\end{itemize}

\item {\bf If the length of $\bz'$ is $(N+1)$}, 
we conclude that a single insertion has occurred. Suppose $\bz'=(z'_1,\ldots,z'_{N+1})$. The decoder proceeds as follows.
\begin{itemize}
\item {\bf Localizing the insertion.} If $z'_{n+1}$ and $z'_{n+2}$ are identical, the decoder sets $\by$ as the first $n$ symbols of $\bz'$. On the other hand, if $z'_{n+1}$ and $z'_{n+2}$ are different, the decoder sets $\by'$ as the first $(n+1)$ symbols of $\bz'$ and there is no error in $\bR_1, \bR_2, \bR_3,$ and $\bR_4$. 
\item {\bf Recovering $\by$.} The decoder follows the steps in Lemma~\ref{decodedel} to recover $\by$ from $\by'$.
\end{itemize}
\end{enumerate}
It is easy to see that all the decoding steps run in $O(n)$ time.
\end{proof}

%

\subsection{Edit Error in DNA Storage Channel}\label{sec:ntedit}

There are recent works that characterize the error probabilities by analyzing data from certain experiments 
\cite{exp1, Organick:2018}. 
Specifically, Heckel \et{} \cite{exp1} studied the substitution errors and computed conditional error probabilities for mistaking a certain nucleotide for another.
 They also compared their data with experiments from other research groups and 
 observed that the probabilities of mistaking $\tT$ for a $\tC$ ($\tT\to \tC$) and 
 $\tA$ for a $\tG$ ($\tA\to \tG$) are significantly higher than substitution probabilities.

Motivated by the study, we consider an alternative error model where substitution errors only occur between certain nucleotides. We refer this error as a {\em nucleotide edit}. The edits that we define earlier is referred as {\em general edit}. In the nucleotide-edit model, besides a single deletion, insertion, a substitution happens only when
\begin{equation*}
\tA \rightarrow \{\tC, \tG\} , \tT \rightarrow \{\tC, \tG\}, \tC \rightarrow \{\tA, \tT\}, \text{ and } \tG \rightarrow \{\tA, \tT\}.  
\end{equation*}

Now, recall that the one-to-one correspondence between $\D=\{\tA,\tT,\tC,\tG\}$ and two-bit sequences is:
\[ \tA \leftrightarrow 00,\quad \tT \leftrightarrow 01,\quad\tC \leftrightarrow 10,\quad\tG \leftrightarrow 11.\]

Then a nucelotide-edit (substitution) in a symbol occurs if and only if the corresponding first bit is flipped.

\begin{example}\label{editexample2}
Suppose that $\bsg=\tA\tC\tA\tG\tT\tG$. Suppose a substitution error occurs at the third symbol, changing $\tA$ to $\tC$, and we received $\bsg_1=\tA\tC{\color{red}{\tC}}\tG\tT\tG$.  
On the other hand, suppose a substitution error also occurs at the third symbol, changing $\tA$ to $\tG$, and we received $\bsg_2=\tA\tC{\color{red}{\tG}}\tG\tT\tG$. We then see that there is exactly one substitution in $\bU_\bsg$ which is also at the third symbol.
\begin{center}
\begin{tabular}{ll}
\begin{tabular}{|c|| c | c| c| c| c|c|}
 \hline
 $\bsg$ & $\tA$ & $\tC$ & $\tA$  & $\tG$ & $\tT$ & $\tG$ \\ \hline
 $\bU_{\bsg}$ & 0 & 1 & 0 & 1 & 0 & 1 \\\hline
 $\bL_{\bsg}$ & 0 & 0 & 0 & 1 & 1 & 1\\\hline
 \end{tabular}

\begin{tabular}{|c|| c | c| c| c| c|c|}
 \hline
 $\bsg_1$ & $\tA$ & $\tC$ & ${\color{red}{\tC}}$  & $\tG$ & $\tT$ & $\tG$ \\ \hline
 $\bU_{\bsg}$ & 0 & 1 & {\color{red}{1}} & 1 & 0 & 1 \\\hline
 $\bL_{\bsg}$ & 0 & 0 & 0 & 1 & 1 & 1\\\hline
 \end{tabular}

\begin{tabular}{|c|| c | c| c| c| c|c|}
 \hline
 $\bsg_2$ & $\tA$ & $\tC$ & ${\color{red}{\tG}}$  & $\tG$ & $\tT$ & $\tG$ \\ \hline
 $\bU_{\bsg}$ & 0 & 1 & {\color{red}{1}} & 1 & 0 & 1 \\\hline
 $\bL_{\bsg}$ & 0 & 0 & {\color{red}{1}} & 1 & 1 & 1\\\hline
 \end{tabular}
\end{tabular}
\end{center}
\end{example}

We now constructed order-optimal nucleotide-edit-correcting codes $\C^{\rm nt}(n;a,b,c)$ for DNA storage as follows.
\begin{construction}\label{nucleotide construction} For given $n,r,P>0, P\geq r+1$, $a\in \mathbb{Z}_{n}, b \in \mathbb{Z}_{P}, c \in \mathbb{Z}_{2}$, let $\C^{\rm nt}_{a,b,c}(n,r,P)$ be the set of all DNA strands $\bsg$ of length $n$ satisfying the following constraints.

\begin{itemize}
\item {\bf Upper-array constraints}.
\begin{itemize}
\item The upper-array $\bU_\bsg$ is a codeword in ${\rm L}_a(n)$. 
\item The longest run of $0$'s or $1$'s in $\bU_\bsg$ is at most $r$. 
\end{itemize}
\item {\bf Lower-array constraint.} The lower array $\bL_\bsg$  is a codeword in ${\rm SVT}_{b,c,P}(n)$.
\end{itemize}
\end{construction}
\begin{theorem}
The code $\C_{a,b,c}^{\rm nt}(n,r,P)$ corrects a single nucleotide edit in linear time.
\end{theorem}
\begin{proof}
Suppose $\bsg'$ is the received word.
According to Proposition~\ref{prop:obs}, $\bsg'\in \Bedit(\bsg)$ implies that $\bU_{\bsg'}\in \Bedit(\bU_\bsg)$ and $\bL_{\bsg'}\in \Bedit(\bL_\bsg)$. Since $\bU_\bsg \in{\rm L}_a(n)$, where ${\rm L}_a(n)$ can correct a single edit, we can recover $\bU_\bsg$ in linear time. 
For $\bL_\bsg$, we have two cases.
\begin{itemize}
\item If $\bsg'$ is of length $n-1$ (or $n+1$), then a deletion or insertion has occurred.
Now, we are able to identify the location of deletion (or insertion) in $\bU_\bsg$, which belongs to a run of length at most $r$. Hence, we can locate the error in $\bL_\bsg$ within $(r+1)$ positions. 
Since $\bL_\bsg \in {\rm SVT}_{b,c,P}(n)$ with $P\geq r+1$, we are able to recover uniquely $\bL_\bsg$.
\item If $\bsg'$ is of length $n$, then a nucleotide substitution has occurred. 
Therefore, it is necessary that a bit flip occur in the $\bU_\bsg$ and hence, we can locate the error.
To correct the corresponding position in $\bL_\bsg$, we simply make use of the syndrome in ${\rm SVT}_{b,c,P}(n)$.\qedhere
\end{itemize}
\end{proof}

Next, we present an efficient encoder that maps messages into $\C_{a,b,c}^{\rm nt}(n,r,P)$.

\noindent{\bf Nucelotide-Edit-Encoder}. Given $n,r=2\ceil{\log n}+4,P= r+1, a \in \mathbb{Z}_{n}, b \in \mathbb{Z}_{P}, c\in \mathbb{Z}_{2}$, set $t\triangleq\ceil{\log P}+1$ and $m\triangleq 2n-\ceil{\log n}-t-2$.

{\sc Input}: $\bx\in \{0,1\}^m$\\
{\sc Output}: $\bsg \triangleq \enc_{\mathbb{E}}^{{\rm nt}}(\bx)\in \C_{a,b,c}(n,r,P)$\\[-2mm]
\begin{enumerate}[(I)]
\item Set $\bx_1$ be the first $(n-\ceil{\log n}-2)$ bits in $\bx$ and $\bx_2$ be the last $(n-t)$ bits in $\bx$.
\item Let $\by_1'=\enc_{RLL}(\bx_1)$ and use Encoder $\mathbb{L}$ as described in Subsection~\ref{template} to encode $\by_1= \enc_{\mathbb{L}}(\by_1') \in {\rm L}_a(n)$.
\item Use SVT-Encoder $\enc_{SVT}$ to encode $\by_2=  \enc_{SVT}(\bx_2)\in {\rm SVT}_{b,c,P}(n)$. 
\item Finally, set $\by=\by_1\by_2$ and output $\bsg\triangleq \Psi^{-1}(\by)$.
\end{enumerate}

\begin{theorem} The Nucelotide-Edit-Encoder is correct. In other words, $\enc_{\mathbb{E}}^{{\rm nt}}(\bx)\in \C_{a,b,c}(n,r,P)$ for all $\bx\in\{0,1\}^m$. The redundancy of our encoder is $\log n+\log \log n + O(1)$.
\end{theorem}

\begin{proof}
Let $\bsg\triangleq \enc_{\mathbb{E}}^{{\rm nt}}(\bx)$. 
Based on our encoder, $\bU_\bsg=\by_1$ and $\by_1= \enc_{\mathbb{L}}(\by_1') \in {\rm L}_a(n)$. In addition, $\by_1'=\enc_{RLL}(\bx_1)$, which implies the longest run of $0$'s or $1$'s in $y_1'$ is at most 
$(\ceil{\log n}+3)$ according to Theorem~\ref{encoderRLL}. 
Therefore, the maximum run in $\bU_\bsg$ after $\enc_{\mathbb{L}}$ is at most $(\ceil{\log n}+3)+(\ceil{\log n}+1)=2\ceil{\log n}+4=r$. Since $P= r+1$, it satisfies the upper-array constraints in Construction~\ref{nucleotide construction}. On the other hand, based on our encoder, $\bL_\bsg=\by_2$ and $\by_2=  \enc_{SVT}(\bx_2)\in {\rm SVT}_{b,c,P}(n)$. It implies $\bL_\bsg$ also satisfies the lower-array constraint in Construction~\ref{nucleotide construction}. Therefore, $\enc_{\mathbb{E}}^{{\rm nt}}(\bx)\in \C_{a,b,c}(n,r,P)$ for all $\bx\in\{0,1\}^m$. The redundancy of our encoder is $\ceil{\log n}+2+\ceil{\log(2\ceil{\log n}+5)}+1=\log n+\log \log n + O(1)$.
\end{proof}

For completeness, we state the corresponding Nucleotide-Edit-Decoder, $\dec_{\mathbb{E}}^{{\rm nt}}(\bsg)$, for DNA codes that correct a single nucleotide edit.

\vspace{2mm}

\noindent{\bf Nucelotide-Edit-Decoder}. For given $n,a,b,c,$ $r=2\ceil{\log n}+4,P=r+1,t=\ceil{\log P}+1$, set $m= 2n-\ceil{\log n}-t-2$.

{\sc Input}: $\bsg\in \D^{n*}$\\
{\sc Output}: $\bx=\dec_{\mathbb{E}}^{{\rm nt}}(\bsg) \in\{0,1\}^m$\\[-2mm]
\begin{enumerate}[(I)]
\item Compute $\hat{\by_1}=\bU_\bsg$ and $\hat{\by_2}=\bL_\bsg$.
\item Compute $\by_1'\triangleq \dec^L_a(\hat{\by_1})$ to correct a single edit in $\hat{\by_1}$.
\item Compute $\by_2\triangleq \dec_{c,d,P}^{SVT}(\hat{\by_2})$ of length $(n-t)$.
\item Compute $\by_1= \dec_{RLL}(\by_1')$ of length $(n-\ceil{\log n}-2)$.
\item Output $\by_1\by_2$ of length $m= 2n-\ceil{\log n}-t-2$.
\end{enumerate}

\section{\GC-Balanced Encoder Correcting Single Edit}
\label{sec:balanced}

We modify the single-edit-correcting Encoder $\mathbb{L}$ in Section~\ref{sec:indel} to obtain a \GC-balanced single-edit-correcting encoder.
Our modification makes use of the celebrated Knuth's balancing technique \cite{knu1986}.

Knuth's balancing technique is a linear-time algorithm that maps a binary message $\bx$
to a balanced word $\bz$ of the same length by flipping the first $k$ bits of $\bx$.
The crucial observation demonstrated by Knuth is that such an index $k$ always exists and 
$k$ is commonly referred to as the {\em balancing index}.
Formally, we have the following theorem.

\begin{theorem}[Knuth \cite{knu1986}]
There exists a pair of linear-time algorithms $\enc^{K}:\{0,1\}^n\to \{0,1\}^n\times[n]$ and 
$\dec^{K}:\{0,1\}^n\times [n]\to \{0,1\}^n$ such that the following holds.
If $\enc^{K}(\bx)=(\bz,k)$, then $\bz$ is balanced and $\dec^K(\bz,k)=\bx$.
\end{theorem}

To represent the balancing index, Knuth appends $\bz$ with a short balanced suffix of length $\ceil{\log n}$ and 
so, a lookup table of size $n$ is required. 
In constrast, since we only require the upper sequence $\bU_\bsg$ to be balanced,
we simply store the balancing index $k$ in the lower sequence $\bL_\bsg$.
Consequently, we do not need a look up table for the balancing indices.

\vspace{1mm}

\noindent{\bf \GC-Balanced Encoder}. Given $n$, set $t=\ceil{\log n}$ and $m=2n-3t-2$.

{\sc Input}: $\bx\in \{0,1\}^n$, $\by\in \{0,1\}^{n-3t-2}$ and so, $\bx\by\in\{0,1\}^m$\\
{\sc Output}: $\bsg = \enc_{\tG\tC}(\bx\by)$\\[-3mm] 

\begin{enumerate}[(I)]
\item Apply Knuth's balancing technique to obtain $(\bz,k)\triangleq \enc^K(\bx)$.
\item Compute $d \triangleq {\rm Syn}(\bz) \ppmod{2n}$ and let $\bd$ be the binary representation of $d$ of length $t+1$.
On the other hand, let $\bk$ be the binary representation of the balancing index $k$ of length $t$.
\item Next, we append $\bd$ and $\bk$ to $\by$ to obtain a binary word of length $n-t-1$.
Using Encoder $\mathbb{L}$, we compute $\bc\triangleq\enc_{\mathbb{L}}(\by\bd\bk)$.
\item Finally, we set $\bsg\triangleq \Psi^{-1}(\bz||\bc)$.
\end{enumerate}

\begin{example}Consider $n=16$ and $a=0$. So,  $t=4$ and $m=2n-3t-2=18$.
Let $\bx= 1111 1111 0000 1111$ and $\by = 01$, and so the message is
$1111 1111 0000 1111 01$.
\begin{enumerate}[(I)]
\item Knuth's method yields $\bz= 0000 1111 0000 1111$ with index $k =4$. 
\item So, $d={\rm Syn}(\bz)= 20 \ppmod{2n}$.  The binary representations of $d$ and $k$ are
$\bd = 10100$  and $\bk = 0100$, respectively.
\item Hence, we have $\by\bd\bk= 01 10100 0100$ and use Encoder 1 to compute
$\bc=\enc_{\mathbb{L}}(\by\bd\bk)= 11 0 1 110 1 1000100 1$.
\item So, the DNA codeword is $\bsg={\tt TTAT GGCG TAAA GCCG}$.
\end{enumerate}
\end{example}

To demonstrate that the map $\enc_{\tG\tC{\rm balanced}}$ corrects a single edit,
we provide an explicit decoding algorithm. 
\vspace{1mm}

\noindent{\bf \GC-Balanced Decoder}. For any $n$, set $m=2n-3\ceil{\log n}-2$.

{\sc Input}: $\bsg\in \D^{n*}$\\
{\sc Output}: $\bx\by=\dec_{\tG\tC}(\bsg) \in\{0,1\}^m$\\[-2mm]
\begin{enumerate}[(I)]
\item Compute $\hat{\bz}=\bU_\bsg$ and $\hat{\bc}=\bL_\bsg$.
\item Compute $\bc\triangleq \dec^L_a(\hat{\bc})$ to correct a single edit in $\hat{c}$.
\item Remove the suffices $\bd$ and $\bk$ to obtain $\by$ and figure out the syndrome $d$ and balancing index $k$.
\item Using the syndrome $d$, compute $\bz\triangleq\dec^L_d(\hat{\bz})$ to correct a single edit in $\hat{\bz}$.
\item Using the index $k$, compute $\bx\triangleq\dec^K(\bz,k)$.
\end{enumerate} 

\begin{theorem}
The map $\enc_{\tG\tC}$ is a $\tG\tC$-balanced single-edit-correcting encoder with redundancy
$3\ceil{\log n}+2$.
\end{theorem}

\begin{proof}
The $\tG\tC$-Balanced Decoder shows that the $\tG\tC$-Balanced Encoder corrects a single edit. 
Hence, it remains to verify that any codeword $\bsg=\enc_{\tG\tC{\rm balanced}}(\bx\by)$ is $\tG\tC$-balanced.
This follows from the fact that $\bU_\bsg=\bz$ is the binary balanced word obtained from balancing $\bx$.
\end{proof}

%
%

\section{Conclusion}
We provide order-optimal quaternary encoders for codes that can correct a single either indel or edit. The encoders map binary messages into quaternary codes in linear-time and the redundancy is at most $\log n+O(\log \log n)$. Moreover, in the case for indel, our encoder uses only $\lceil\log{n}\rceil +2$ redundant bits,  improving the best known encoder of Tenengolts (1984) by at least four bits. We also present linear-time encoders for \GC-balanced codes that can correct a single edit. 
The redundancy of this encoder is $3\lceil\log{n}\rceil+2$.
\vspace{3mm}


\end{document}